\documentclass[10pt]{article}
\pdfoutput=1
\usepackage{epic,graphpap,graphics,float,tabularx}
\usepackage{paralist,longtable,fancyhdr,fancyvrb,bbm,natbib}
\usepackage{amsmath}
\usepackage{amstext}
\usepackage{amsbsy}
\usepackage{amsthm}
\usepackage{amsfonts}
\usepackage{amssymb}
\usepackage{amscd}
\usepackage{eucal}

\theoremstyle{plain}
\newtheorem{thm}{Theorem}[section]

\newtheorem{cor}[thm]{Corollary}
\newtheorem{prop}[thm]{Proposition}
\theoremstyle{definition}
\newtheorem{defn}[thm]{Definition}
\newtheorem{exa}[thm]{Example}

\numberwithin{equation}{section}
\numberwithin{equation}{section}

\def\as{\quad\text{{\rm a.s.}}}
\def\sumi{\sum_{i=1}^n}
\def\sumj{\sum_{j=1}^n}
\def\sumk{\sum_{k=1}^n}

\def\m{\mu}

\def\t{\tau}
\def\p{\pi}

\def\g{\gamma}

\def\s{\sigma}

\def\D{\Delta}
\def\F{{\mathcal F}}
\def\FF{{\widetilde{F}}}

\def\S{{\mathbf S}}
\def\ph{\varphi}

\def\1{{\mathbbm 1}}

\def\N{{\mathbb N}}

\def\R{{\mathbb R}}

\def\BX{{\mathbf X}}
\def\L{\Lambda}
\def\eqdef{\triangleq}

\def\bs{\boldsymbol \sigma}
\def\bl{\boldsymbol \lambda}
\def\l{\lambda}
\def\intt{\int_0^t}
\def\sumi{\sum_{i=1}^n}
\def\sumij{\sum_{i,j=1}^n}

\def\half{\frac{1}{2}}
\def\brac#1{\langle #1\rangle}
\def\bbrac#1{\big\langle #1 \big\rangle}
\def\lt#1{\Lambda_{#1}}
\def\ito{It\^o}
\def\limT#1{\lim_{T\to\infty}\frac{#1}{T}}
\def\limt#1{\lim_{t\to\infty}\frac{#1}{t}}
\def\intt{\int_0^t}
\def\intT{\int_0^T}
\def\dd{\circ d}
\def\XX{\{X_1,\ldots,X_n\}}
\def\BXX{\{{\mathbf X_1},\ldots,{\mathbf X_n}\}}

\def\ito{It{\^o}}
\def\S{{\bf S}}
\def\m{\mu}
\def\D{\Delta}
\def\p{\pi}
\def\g{\gamma}
\def\as{\quad\text{{\rm a.s.}}}
\def\R{{\mathbb R}}
\def\eqdef{\triangleq}
\def\half{\frac{1}{2}}
\def\sumi{\sum_{i=1}^n}
\def\sumj{\sum_{j=1}^n}
\def\brac#1{\langle #1 \rangle}
\def\dd{\circ d}
\def\T{{\mathcal {T}}}
\def\Scal{{\mathcal {S}}}
\def\intT{\int_0^T}
\def\sumij{\sum_{i,j=1}^n}

\begin{document}
\centerline{\bf \LARGE{Permutation-weighted portfolios and the efficiency}}
\vspace{5pt}
\centerline{\bf \LARGE{of commodity futures markets}}
%\blfootnote{The authors thank xxx for their invaluable comments and suggestions regarding this research.}

\vskip 35pt

\centerline{\large{Ricardo T. Fernholz\footnote[1]{Claremont McKenna College, 500 E. Ninth St., Claremont, CA 91711, rfernholz@cmc.edu.} \hskip 60pt Robert Fernholz\footnote[2]{Intech Investment Management, LLC, One Palmer Square, Suite 441, Princeton, NJ 08542, bob@bobfernholz.com.}}}

\vskip 35pt

\centerline{\large{\today}}

\vskip 60pt

\begin{abstract}
A {\em market portfolio} is a portfolio in which each asset is held at a weight proportional to its market value. Functionally generated portfolios are portfolios for which the logarithmic return relative to the market portfolio can be decomposed into a function of the market weights and a process of locally finite variation, and this decomposition is convenient for characterizing the long-term behavior of the portfolio. A {\em permutation-weighted portfolio} is a portfolio in which the assets are held at weights proportional to a permutation of their market values, and such a portfolio is functionally generated only for markets with two assets (except for the identity permutation). A {\em reverse-weighted portfolio} is a portfolio in which the asset with the greatest market weight is assigned the smallest market weight, the asset with the second-largest weight is assigned the second-smallest, and so forth. Although the reverse-weighted portfolio in a market with four or more assets is not functionally generated, it is still possible to characterize its long-term behavior using rank-based methods. This result is applied to a market of commodity futures, where we show that the reverse price-weighted portfolio substantially outperforms the price-weighted portfolio from 1977-2018.
\end{abstract}

\vskip 205pt


\vskip 5pt


\vfill

\section{Introduction} \label{intro}%%%%%%%%%%%%%%%%%%%

Functionally generated portfolios are portfolios with weights derived from a positive $C^2$ function of the market weights. These portfolios were introduced by \citet{F:pgf}, and there have since been a number of generalizations \citep{strong:2014,KR:2017,KK:2020}. Functionally generated portfolios can be constructed to outperform a capitalization-weighted stock market portfolio under realistic conditions \citep{F:2002, FKK:2005, FKR:2018}. We show that a surprisingly simple portfolio, a {\em swap portfolio}, will almost surely outperform the market portfolio over the long term under the weak condition of market {\em coherence}. A swap portfolio is functionally generated, and holds only two assets with the weight of each proportional to the market weight of the other.

%There have been a number of generalizations of functionally generated portfolios, e.g., \citet{strong:2014}, \citet{KR:2017}, and \citet{KK:2020}, however here we consider a simple example, a {\em swap portfolio}. A swap portfolio holds only two assets, and the weight of each is proportional to the market weight of the other. A swap portfolio is functionally generated, and under the weak condition of market {\em coherence,} i.e., that all the assets in the market have the same asymptotic growth rate, we show that a swap portfolio will almost surely outperform the market portfolio over the long term. 

The weights of a functionally generated portfolio are functions of the market weights, but the converse is not true. Indeed, a {\em permutation-weighted portfolio}, in which the assets are held at weights proportional to a permutation of their market values, will not be functionally generated if the market contains more than two assets (except for the identity permutation). We extend the decomposition of functionally generated portfolios to general portfolios with weight functions that are continuous semimartingales using the Stratonovich integral, although it is not possible to characterize the behavior of the components of this general decomposition in the same manner as for functionally generated portfolios.

We next consider rank-based portfolios and, specifically, the {\em reverse-weighted portfolio}, in which the weights of the market portfolio are reversed according to rank. {\em Atlas models} and {\em first-order models} are systems of continuous semimartingales with coefficients that depend only on rank \citep{F:2002,BFK:2005,IPBKF:2011}. The mathematical theory of Atlas models and first-order models developed in \citet{BFK:2005} and \citet{IPBKF:2011} is based on a number of earlier results. The existence and uniqueness for solutions of these systems comes from \citet{Bass/Pardoux:1987} and \citet{Stroock/Varadhan:2006}. The long-term behavior of Atlas models and first-order models, including the existence of a stationary distribution and a strong law of large numbers, can be found in \citet{Khasminskii:1960,Khasminskii:1980}. The theory of rank-based systems of continuous semimartingales has been extended in several directions, e.g., infinite Atlas systems \citep{Pal/Pitman:2008,Chatterjee/Pal:2013,Bruggeman:2016}; behavior at triple points \citep{Banner/Ghomrasni:2008}; existence and nonexistence of triple points \citep*{IK:2010,IchibaKaratzasShkolnikov:2013,Sarantsev:2015}; convergence to equilibrium \citep*{IchibaPalShkolnikov-2013,Dembo/Jara/Olla:2017,Dembo/Tsai:2017}; behavior of degenerate systems \citep*{FIK:2013b,FIKP:2013}; large deviations \citep{DSVZ:2016}; and second-order stock-market models \citep*{FIK:2013a}.

We analyze the behavior of a reverse-weighted portfolio in a market represented by a first-order model, and we show that in such a market the reverse-weighted portfolio will almost surely grow faster than the market. We apply this result to a market of commodity futures, a market that can be approximated by a first-order model with rank-symmetric variance parameters. For our application to commodities, we construct implied two-month futures prices and then normalize by setting these prices to be the same on the starting date for the data in a manner similar to \citet{Asness/Moskowitz/Pedersen:2013}. We show that the first-order model for implied two-month commodity futures prices from 1995-2018 has rank-symmetric variance parameters and growth rate parameters that are substantially lower at top ranks than at bottom ranks. These estimated parameters are similar to the first-order parameters estimated for spot commodity prices by \citet{Fernholz:2017a}. 

Consistent with our theoretical results, we show that the reverse-weighted portfolio of commodity futures outperforms the price-weighted market portfolio of commodity futures from 1977-2018. We also show that  over this same time period the reverse-weighted portfolio outperforms the diversity-weighted portfolio with a negative parameter \citep{Vervuurt/Karatzas:2015} as well as the equal-weighted portfolio of commodity futures. These results point to an inefficiency in the commodity futures market.

%The outperformance of the reverse-weighted portfolio relative to the price-weighted market portfolio points to an inefficiency in the commodity futures market. In this market the change in the implied two-month commodity futures prices is not necessarily equal to the return from holding the underlying futures contract. The difference between price changes and returns is the carry. Our results explain the higher growth rate of implied commodity futures prices in the reverse-weighted portfolio relative to the price-weighted portfolio, but not the higher returns of the reverse-weighted portfolio relative to the price-weighted portfolio. In an efficient market, the differential carry would offset the difference in growth rates, but in the commodity futures market the differential carry is not sufficient to equate the returns for the reverse-weighted and price-weighted portfolios. This results in an inefficiency.

\section{Markets and Market Portfolios} \label{markets}%%%%%%%%%%%%%%%%%%%
 
In this section we introduce some of the basic ideas of stochastic portfolio theory, and further details can be found in \citet{F:2002} and \citet{FK:2009}.  For $n\ge2$, consider a market represented by a family $\XX$ of strictly positive continuous semimartingales with the usual filtration $\F^X_t$, for $t\in[0,\infty)$, such that $X_i(t)$ represents the market value of the $i$th asset at time $t\ge0$. Let us specifically note that here we are defining the $X_i$ as the asset {\em values} rather than as stock {\em capitalizations}, which is the usual practice in stochastic portfolio theory (see, e.g.,  \citet{F:2002}). The reason for this deviation from common practice is that we wish to apply our results to commodities markets, where prices or values exist, but capitalizations have no meaning. Since both commodity prices and stock capitalizations are strictly positive, the usual theory can be applied to either.

Let $\p$ be a portfolio with weight processes $\p_1,\ldots,\p_n$, which are bounded measurable processes adapted to $\F^X$, such that $\p_1(t)+\cdots+\p_n(t)=1$ for $t\ge0$. For a portfolio $\p$, the {\em portfolio value process} $Z_\p$ will satisfy
\[
dZ_\p(t) \eqdef Z_\p(t)\sumi\p_i(t)\frac{dX_i(t)}{X_i(t)},
\]
or, in logarithmic terms,
\begin{equation}\label{1.1}
d\log Z_\p(t)=\sumi\p_i(t)\,d\log X_i(t)+\g^*_\p(t) \, dt,\as,
\end{equation}
with the {\em excess growth rate process}
\begin{align}
\g^*_\p(t)&\eqdef\half\bigg(\sumi\p_i(t)\s_{ii}(t)-\s^2_\p(t)\Big) \label{1.2}\\
&=\half\bigg(\sumi\p_i(t)\s_{ii}(t)-\sumij\p_i(t)\p_j(t)\s_{ij}(t)\bigg),\as,\label{1.101}
\end{align}\
where
\begin{align*}
\s_{ij}(t) \, dt &\eqdef d\brac{\log X_i,\log X_j}_t,\\
\s^2_\p(t) \, dt &\eqdef d\brac{\log Z_\p}_t,
\end{align*}
$\brac{\,\cdot\,}_t$ represents the quadratic variation process, and $\brac{\,\cdot,\cdot\,}_t$ represents the cross variation process. It can be shown that 
\begin{equation*} %\label{1.3}
\g^*_\p(t)\ge0,\as,
\end{equation*}
if the $\p_i(t)\ge0$, for $i=1,\ldots,n$, and this provides a measure of the efficacy of diversification in the portfolio. 

Let us denote the total value of the market by $X(t)=X_1(t)+\cdots+X_n(t)$. The {\em market portfolio} $\m$ is the portfolio with weights $\m_1,\ldots,\m_n$ such that each asset is weighted proportionally to its market value:
\begin{equation*}
 \m_i(t)=X_i(t)/X(t).
\end{equation*}
It can be shown that, with appropriate initial conditions, the value process of the market portfolio satisfies $Z_\m(t)=X(t)$. The {\em relative covariance processes} are defined by
\[
\t_{ij}(t)\,dt\eqdef d\brac{\log\m_i,\log\m_j}_t,
\]
and the excess growth rate process can be expressed as
\[
\g^*_\p(t)=\half\bigg(\sumi\p_i(t)\t_{ii}(t)-\sumij\p_i(t)\p_j(t)\t_{ij}(t)\bigg),\as
\]

For a portfolio $\p$, the portfolio log-return relative to the market satisfies
\begin{equation}\label{1.4}
d\log \big(Z_\p(t)/Z_\m(t)\big)=\sumi\p_i(t)\,d\log \m_i(t)+\g^*_\p(t) \, dt,\as
\end{equation}
(see \citet{F:2002}, Proposition 1.2.5). In the case that $\p=\m$, the left-hand side vanishes, and we have
\begin{equation}\label{1.5}
\sumi\m_i(t)\,d\log \m_i(t)=-\g^*_\m(t) \, dt\le0,\as
\end{equation}
From this we see that whatever benefit the market has from diversification is lost in the weighted average of the $\log \m_i$ terms. This suggests that market weights might not always be ``optimal'', and some kind of improvement may be possible.

In order to understand the long-term behavior of portfolios, we need to impose some asymptotic stability conditions. The market is {\em coherent} if for $i=1,\ldots,n$,
\begin{equation*} %\label{1.6}
\limt{1} \log\m_i(t) = 0,\as
\end{equation*}
The market is {\em asymptotically diversified} if 
\begin{equation*} %\label{1.7}
\limT{1} \intT \g_\m^*(t)\,dt > 0, \as
\end{equation*}
The market is  {\em pairwise asymptotically diversified} if for each $1\le i\ne j\le n$, the submarket $\{X_i,X_j\}$ is asymptotically diversified.

\section{Permutation-weighted portfolios} \label{perm}%%%%%%%%%%%%%%%%

In this section we consider functionally generated portfolios and permutation-weighted portfolios. Functionally generated portfolios were introduced by \citet{F:pgf,F:2002}, and we define permutation-weighted portfolios below.

A positive $C^2$ function $\S$ defined on the unit simplex $\D^n\subset\R^n$ {\em generates} a portfolio  $\p$ if 
\begin{equation}\label{2.1}
\log\big(Z_\p(t)/Z_\m(t)\big)  = \log \S(\m(t))  + \Theta(t),\as,
\end{equation}
where the {\em drift process} $\Theta$ is of locally bounded variation. It was shown by \citet{F:2002}, Theorem 3.1.5, that the portfolio $\p$ will have weights
\begin{equation}\label{2.2}
\p_i(t)=\Big(D_i\log\S(\m(t))+1-\sumj\m_j(t)D_j\log\S(\m(t))\Big)\m_i(t),\as,
\end{equation}
for $i=1,\ldots,n$, with 
\begin{equation}\label{2.3}
d\Theta= \frac{-1}{2\S(\m(t))}\sumij D_{ij}\S(\m(t))\m_i(t)\m_j(t)\t_{ij}(t) \, dt,\as
\end{equation}

For $n\ge2$, we see that \eqref{2.2} and \eqref{2.3} indicate that for $1\le i<j\le n$, the function
\begin{equation*} %\label{2.4}
\S(x)=\frac{x_i x_j}{x_i+x_j},
\end{equation*}
generates the {\em swap portfolio}  $\p$  with weight processes 
\begin{equation}\label{2.40}
\p_i(t)=\frac{\m_j(t)}{\m_i(t)+\m_j(t)} \quad\text{ and }\quad \p_j(t)=\frac{\m_i(t)}{\m_i(t)+\m_j(t)},\as,
\end{equation}
and $\p_k = 0$ for $k \notin \{i, j\}$, and with drift process
\begin{equation}\label{2.5}
\Theta(t)= \intt \frac{\m_i(s)\m_j(s)}{(\m_i(s)+\m_j(s))^2}\big(\t_{ii}(s)-2\t_{ij}(s)+\t_{jj}(s)\big)ds,\as
\end{equation}
We see that the weights $\p_i,\p_j$ are proportional to ``swapped'' market weights $\m_i,\m_j$. We would like to see if this might improve on the negative effect seen in \eqref{1.5}.

\begin{prop}\label{P1} Suppose that for $n\ge2$ the market $\{X_1,\ldots,X_n\}$ is coherent and pairwise asymptotically diversified. Then the swap portfolio $\p$ with weight processes $\p_i$ and $\p_j$ as in \eqref{2.40} will have a higher asymptotic growth rate than the market portfolio.
\end{prop}
\begin{proof}
For $1\le i<j\le n$, let
\[
\eta_i=\frac{\m_i}{\m_i+\m_j}\quad\text{ and }\quad \eta_j=\frac{\m_j}{\m_i+\m_j},
\]
with $\eta_k = 0$ for $k \notin \{i, j\}$, be the weight processes for the passive portfolio holding $\{X_i,X_j\}$. We see from \eqref{2.5} that the drift process satisfies
\begin{align*}
d\Theta(t)&= \frac{\m_i(t)\m_j(t)}{(\m_i(t)+\m_j(t))^2}\big(\t_{ii}(t)-2\t_{ij}(t)+\t_{jj}(t)\big)dt\\
&=  \eta_i(t)\eta_j(t)\big(\t_{ii}(t)-2\t_{ij}(t)+\t_{jj}(t)\big)dt\\
&=  \eta_i(t)\eta_j(t)\big(\s_{ii}(t)-2\s_{ij}(t)+\s_{jj}(t)\big)dt\\
&=  \big(\eta_i(t)\s_{ii}(t)+\eta_j(t)\s_{jj}(t)-\s^2_\eta(t)\big)dt\\
&=  2\g_\eta^*(t) \, dt,\as,
\end{align*}
using (1.2.3) from \citet{F:2002}. Hence, 
\begin{equation*} %\label{2.50}
\log\big(Z_\p(t)/Z_\m(t)\big)  = \log \m_i(t)+ \log\m_j(t) -\log\big(\m_i(t)+\m_j(t)\big)+ 2\intt\g^*_\eta(s) \, ds,\as
\end{equation*}
Since the market $\{X_1,\ldots,X_n\}$ is coherent and  pairwise  asymptotically diversified, 
\begin{align*}
\limt{1}\log\big(Z_\p(t)/Z_\m(t)\big)
 &=2\limt{1}\intt\g^*_\eta(s) \, ds\notag\\
&>0,\as,
\end{align*}
and the proposition follows.
\end{proof}

What can we say about swap portfolios with more than two assets? Unfortunately, it appears that we cannot say much because such portfolios are not functionally generated. Let us now consider portfolios in which an arbitrary number of market weights have been interchanged.

In a market $\{X_1,\ldots,X_n\}$, for $n\ge2$, a {\em permutation-weighted portfolio} is a portfolio $\p$ with weights $\p_i(t)=\m_{p(i)}(t)$, for $i=1,\ldots,n$, where $p$ is a permutation $p\in\Sigma_n$, the symmetric group on $n$ elements. It is well known that a permutation $p$ can be factored into disjoint cycles, each of which represents a mapping of the form 
\begin{equation}\label{6.1}
i_1\mapsto i_2 \mapsto i_3 \mapsto \cdots \mapsto i_m \mapsto i_1,
\end{equation}
with $p(i_j)=i_{j+1}$, for $j<m$, $p(i_m)=i_1$, where $i_1,i_2,\ldots,i_m$ are $m\le n$ distinct elements of $\{1,2,\ldots,n\}$ (see, e.g., \citet{Stanley:2012}, Section~1.3). It is not difficult to show that these cycles will be unique modulo the choice of starting point $i_1$. 

We see from Proposition~\ref{P1} that for $n=2$, the permutation-weighted portfolio $\p$ with $\p_1=\m_2$ and $\p_2=\m_1$ is functionally generated. For $n\ge2$ the market portfolio is permutation-weighted by the identity permutation, and it is generated by the function $\S(x)=1$ (see Example~3.1.6(1), of \citet{F:2002}). As it happens, for $n>2$, the market portfolio is the only functionally generated permutation-weighted portfolio. We need the following result from \citet{F:2002}, Proposition~3.1.11, to prove this.

\begin{prop}\label{L1} 
Let $\{X_1,\ldots,X_n\}$ be a market with $n\ge2$ and let $U\subset\R^n$ be a neighborhood of $\D^n\subset\R^n$. Suppose that $f:U\to\R$ is a  $C^1$ function  such that $f_1(x)+\cdots+f_n(x)=1$, for $x\in U$, and that $\p$ is a portfolio  with $\p_i(t)=f_i(\m(t))$, for $i=1,\ldots,n$. Then $\p$ is functionally generated if and only if  there exists a $C^1$ function $h$ defined on $U$ such that 
\begin{equation}\label{2.6}
D_j\big(f_i(x)/x_i+h(x)\big)=D_i\big(f_j(x)/x_j+h(x)\big),
\end{equation}
for $i,j=1,\ldots,n$ and $x\in U$.
\end{prop}

\vspace{0pt}

\begin{prop}\label{P2.1} 
For a market $\{X_1,\ldots,X_n\}$ with $n>2$, the only permutation-weighted portfolio that is functionally generated is the market portfolio.
\end{prop}

\begin{proof}  Suppose that $n>2$, that $p\in\Sigma_n$ is a permutation and that $\p$ is the permutation-weighted portfolio with $\p_i=\m_{p(i)}$, for $i=1,\ldots,n$. Suppose that one of the cyclical factors of $p$ is of order 2, so that for some $i\ne j\in\{1,\ldots,n\}$, $p(i)=j$ and $p(j)=i$. Since $n>2$, there are $k,\ell\in\{1,\ldots,n\}$ with $p(k)=\ell$ such that $i\ne k \ne j$ and $i\ne \ell \ne j$. If $\p$ is functionally generated, then Proposition~\ref{L1} implies that there is neighborhood $U$ of $\Delta^n$ and a $C^1$ function $h:U\to \R$ such that
\[
D_j\big(x_j/x_i+h(x)\big)=D_i\big(x_i/x_j+h(x)\big),
\]
or
\begin{equation}\label{2.71}
\frac{1}{x_i}+D_j h(x)= \frac{1}{x_j}+D_ih(x).
\end{equation}
We also have
\[
D_j\big(x_\ell/x_k+h(x)\big)=D_k\big(x_i/x_j+h(x)\big),
\]
so $D_j h(x)=D_k h(x)$, and similarly, $D_i h(x)=D_k h(x)$. Hence $D_i h(x)=D_j h(x)$,
which together with \eqref{2.71} implies that
\[
\frac{1}{x_i}= \frac{1}{x_j},
\]
for all $x\in\Delta^n$, which is a contradiction. Hence, there can be no cyclical factors of order 2.

Suppose now that one of the cyclical factors of $p$ is of the form \eqref{6.1} with $2<m\le n$, so
\[
\p_{i_1}(t)=\m_{i_2}(t),\p_{i_2}(t)=\m_{i_3}(t),\ldots,\p_{i_{m-1}}(t)=\m_{i_m}(t),\p_{i_m}(t)=\m_{i_1}(t),
\]
for $t\ge0$. In this case the function $f:U\to\R$ of the proposition satisfies
\[
f_{i_1}(x)=x_{i_2},f_{i_2}(x)=x_{i_3},\ldots,f_{i_m}(x)=x_{i_1},
\]
so we need to find $h:U\to\R$ as in \eqref{2.6} such that
\[
D_{i_{k+1}}\big(f_{i_k}(x)/x_{i_k}+h(x)\big)=D_{i_k}\big(f_{i_{k+1}}(x)/x_{i_{k+1}}+h(x)\big),\quad\text{ for }k=1,\ldots,m-1,
\]
and
\[
D_{i_{1}}\big(f_{i_m}(x)/x_{i_m}+h(x)\big)=D_{i_m}\big(f_{i_{1}}(x)/x_{i_{1}}+h(x)\big).
\]
From this we have,
\begin{equation}\label{2.72}
D_{i_{k+1}}\big(x_{i_{k+1}}/x_{i_k}+h(x)\big)=D_{i_k}\big(x_{i_{k+2}}/x_{i_{k+1}}+h(x)\big),
\end{equation}
for $k=1,\ldots,m-2$, with
\begin{equation}\label{2.73}
D_{i_{m}}\big(x_{i_m}/x_{i_{m-1}}+h(x)\big)=D_{i_{m-1}}\big(x_{i_{1}}/x_{i_{m}}+h(x)\big),
\end{equation}
and
\begin{equation}\label{2.74}
D_{i_{1}}\big(x_{i_1}/x_{i_{m}}+h(x)\big)=D_{i_{m}}\big(x_{i_{2}}/x_{i_{1}}+h(x)\big).
\end{equation}
If we carry out the differentiations and then add up each side of equations \eqref{2.72}, \eqref{2.73}, and  \eqref{2.74}, the $D_{i_k} h(x)$ terms cancel and we find that
\[
\frac{1}{x_1}+\cdots+\frac{1}{x_m}=0,
\]
for all $x\in\Delta^n$, which is a contradiction. Hence, there can be no cyclical factors of order $m>2$, so only the identity permutation remains, and this generates the market portfolio.
\end{proof}

We would like to characterize the behavior of permutation-weighted portfolios of markets $\{X_1,\ldots,X_n\}$ with $n>2$, but we have now seen that we cannot use portfolio generating functions to do so. It happens that the decomposition \eqref{2.1} can be extended to general portfolios $\p$ with weight functions $\p_i$ that are continuously differentiable, and we consider this in the next section.

\section{Decomposition of portfolio log-return} \label{dec}%%%%%%%%%%%%%%%%%%%

In order to generalize the decomposition \eqref{2.1} to a wider class of portfolios we need to consider both \ito\ integration and Stratonovich integration. Details regarding the relationship between these two forms of stochastic integration can be found in Chapter~V of  \citet{protter:1990}. For  continuous square-integrable semimartingales $X$ and $Y$ the {\em Stratonovich integral,} with the differential denoted by  $\phantom{\!\!}\dd$, is defined by
\begin{equation}\label{4}
\intT Y(t)\dd X(t) \eqdef \intT Y(t)\,dX(t)+ \half \,\brac{X,Y}_T,
\end{equation}
for $T\ge0$, where $\brac{X,Y}_t$ is the cross-variation process for $X$ and $Y$. It will be convenient to replace the integral in  \eqref{4} with the differential notation, so that
\begin{equation}\label{5}
Y(t)\dd X(t) =  Y(t)\,dX(t) + \half \,d\brac{X,Y}_t,\as
\end{equation}
It follows from \eqref{5} that for continuous semimartingales the difference between an \ito\ integral and the corresponding Stratonovich integral will be a process of locally bounded variation.

For a $C^2$ function $F$ defined on the range of $X$, the Stratonovich integral satisfies the rules of standard calculus. By \ito's rule we have
\begin{equation*} %\label{6}
dF(X(t))=  F'(X(t))\,d X(t) + \half  F''(X(t))\,d\brac{X}_t,\as,
\end{equation*}
where $\brac{X}_t$ is the quadratic variation of $X$, and since
\[
d\brac{F'(X),X}_t = F''(X(t))\,d\brac{X}_t,\as,
\]
it follows from \eqref{5} that 
\begin{equation*} %\label{7}
dF(X(t)) = F'(X(t))\dd X(t),\as
\end{equation*}
(see \citet{protter:1990},  Theorem~V.20).

\begin{defn}\label{D2} Let $\p$ be a portfolio with weight functions $\p_i$, for $i=1,\ldots,n$, that are continuous semimartingales. The log-return of the portfolio relative to the market can be decomposed as
\begin{equation}\label{7.1}
 d\log\big(Z_\p(t)/Z_\m(t)\big) =d\log\Scal_\p(t) +d\T_\p(t),\as,
\end{equation}
where the {\em structural process} $\Scal_\p$ of $\p$ is defined by
\begin{equation*} %\label{8}
d\log\Scal_\p(t)\eqdef \sumi \p_i(t)\dd\log \m_i(t),
\end{equation*}
and the {\em trading process} $\T_\p$ is defined by
\begin{equation}\label{9}
d\T_\p(t)\eqdef d\log\big(Z_\p(t)/Z_\m(t)\big)-d\log\Scal_\p(t).
\end{equation}
\end{defn}
\vspace{5pt}

The motivation for the terms ``structural process'' and ``trading process'' can be found in \citet{F:2016}.
\begin{prop} Let $\p$ be a portfolio with weight functions $\p_i$, for $i=1,\ldots,n$, that are continuous semimartingales. Then the trading process $\T_\p$ is of locally bounded variation.
\end{prop}

\begin{proof}  From \eqref{9} we have
\begin{align}
d\T_\p(t)&=d\log\big(Z_\p(t)/Z_\m(t)\big)-d\log\Scal_\p(t)\notag\\
&= \sumi\p_i(t)\,d\log\m_i(t)+\g^*_\p(t)\,dt-\sumi\p_i(t)\dd\log\m_i(t)\notag\\
&= \Big(\sumi\p_i(t)\,d\log\m_i(t)-\sumi\p_i(t)\dd\log\m_i(t)\Big)+\g^*_\p(t)\,dt\notag\\
&= \sumi d\bbrac{\p_i(t),\log\m_i(t)}_t+\g^*_\p(t)\,dt,\as \label{10}
\end{align}
%where $\g^*_\p$ is the excess growth rate of $\p$ (see \citet{F:2002}). 
Since the portfolio weight processes $\p_i$ are bounded continuous semimartingales, the cross-variation terms in \eqref{10} are of locally bounded variation, and since the excess growth term is also of locally bounded variation,  so will $\T_\p$ be.
\end{proof}

We can now show that for functionally generated portfolios the decompositions \eqref{2.1} and \eqref{7.1} coincide.

\begin{prop}\label{P4.3} Let $\p$ be the portfolio generated by the positive $C^2$ function $\S$ with drift process $\Theta$. Then
\[
d\log\Scal_\p(t) = d\log\S(\m(t)),\as,
\]
and
\begin{equation*} %\label{14}
d\T_\p(t)= d\Theta(t),\as
\end{equation*}
\end{prop}

\begin{proof}  Since, for $i=1,\ldots,n$, the terms $D_i\log\S(\m(t))$ are continuous semimartingales, we can apply Stratonovich integration, with
\begin{align}
d\log\S(\m(t))&=\sumi D_i\log\S(\m(t))\dd\m_i(t)\notag\\
&=\sumi D_i\log\S(\m(t))\m_i(t)\dd\log \m_i(t)\notag\\ 
&=\sumi \p_i(t)\dd\log\m_i(t)\label{15} \\
&= d\log\Scal_\p(t),\as,\label{16} 
\end{align} 
by Definition \ref{D2}, where \eqref{15} follows from \eqref{2.2} and the fact that
\[
\sumi \m_i(t)\dd\log\m_i(t)=\sumi d\m_i(t)=d\sumi \m_i(t)=0,\as
\]
With \eqref{16} established, 
\begin{equation*}
d\T_\p(t)= d\Theta(t),\as,
\end{equation*}
follows directly from the definitions \eqref{2.1} and \eqref{9}. 
\end{proof}

Proposition~\ref{P4.3} shows that Definition~\ref{D2} extends the decomposition \eqref{2.1} of the relative log-return for functionally generated portfolios to general portfolios with weights that are continuous semimartingales, and this latter class includes permutation-weighted portfolios. However, while for certain types of generating function $\S$ it is possible to describe the  behavior of the components $\log\S$ and $\Theta$ of a functionally generated portfolio, we know of no simple characterization of the components $\log\Scal_\p$ and $\T_\p$ defined for more general portfolios. 

\begin{exa} Consider the market $\{X_1,\ldots,X_n\}$, with $n>2$, and let $p\in\Sigma_n$ be a permutation other than the identity. Let $\p$ be the permutation-weighted portfolio with $\p_i=\m_{p(i)}$, for $n=1,\ldots,n$. Then we can calculate the components of the decomposition \eqref{7.1}, with
\begin{align}
d\log\Scal_\p(t)&= \sumi \p_i(t)\dd\log \m_i(t)\notag\\
&=\sumi \m_{p(i)}(t)\dd\log \m_i(t),\as,\label{p6-1}
\end{align}
and, 
\begin{align*}
d\T_\p(t)&= d\log\big(Z_\p(t)/Z_\m(t)\big)-d\log\Scal_\p(t)\notag\\
&= \sumi\m_{p(i)}(t)\,d\log\m_i(t)+\g_\p^*(t)\,dt-\sumi \m_{p(i)}(t)\dd\log \m_i(t)\notag\\
\end{align*}

\vspace{-12pt}
\noindent We can now apply \eqref{4} to obtain
\begin{align}
d\T_\p(t)&= -\half\sumi d\brac{\m_{p(i)},\log\m_i}_t+\g_\p^*(t)\,dt\notag\\
&= -\half\sumi\m_{p(i)}(t)\t_{p(i) i}(t)\,dt+\g_\p^*(t)\,dt,\as\label{p6-2}
\end{align}
\end{exa}

The expressions in \eqref{p6-1} and \eqref{p6-2} do not appear to be amenable to further analysis, so although the decomposition \eqref{7.1} can be applied to the permutation-weighted portfolio $\p$, this does not allow us to extend Proposition~\ref{P1} to include  markets with $n>2$. Hence, to understand the behavior of permutation-weighted portfolios in markets with more than two assets, we must rely on other techniques. Accordingly, in the next section we consider {\em rank-based} methods in which the market assets are studied in terms of rank rather than name, or index. 

\section{First-order approximations to asymptotically stable markets} \label{fom}%%%%%%%%%%%%%%%%%%%

In this section we consider portfolio behavior in terms of rank, so we shall introduce the basic concepts of rank-based analysis for systems of continuous semimartingales (see, e.g., \citet{F:2002}). To begin, for $t \in [0,\infty)$, let $r_t\in\Sigma_n$, the symmetric group on $n$ elements, be the rank function for $X_1(t),\ldots,X_n(t)$, with $r_t(i)<r_t(j)$ if $X_i(t)>X_j(t)$ or if $X_i(t)=X_j(t)$ and $i<j$. The corresponding {\em rank processes} $X_{(1)}\ge\cdots\ge X_{(n)}$  are defined by $X_{(r_t(i))}(t)=X_i(t)$. We have assumed that the semimartingales $X_i(t)$ are strictly positive, so we can consider the logarithmic processes $\log X_1,\ldots,\log X_n$. For $1\le k< \ell \le n$, let $\lt{k,\ell}^X$ denote  the local time at the origin for $\log X_{(k)}-\log X_{(\ell)}$, with $\lt{0,1}^X=\lt{n,n+1}^X\equiv 0$ (see \citet{Karatzas/Shreve:1991}, Section~3.7).

We are interested in the asymptotic behavior of portfolios relative to the market, so it is reasonable that we restrict our attention to markets with some level of asymptotic stability. 

\begin{defn} \label{D4.1}
\citep{F:2002} The family  $\XX$  of strictly positive continuous square-integrable semimartingales is {\em asymptotically stable} if
\begin{enumerate}
\item $\displaystyle
     \limt{1}\big( \log X_{(1)}(t)-\log X_{(n)}(t)\big)=0,\as$ ({\em coherence});
\item $\displaystyle
  \l_{k,k+1}    =       \limt{1}\L_{k,k+1}^X(t) >0,\as$;
\item $\displaystyle
\s^2_{k,k+1}=      \limt{1}\bbrac{\log \big(X_{(k)}/X_{(k+1)}\big)}_t >0,\as$;
\end{enumerate}
 for $k=1,\ldots,n-1$, where $\L_{k,k+1}^X$ is the local time at the origin for $\log\big(X_{(k)}/X_{(k+1)}\big)$. By convention, let $\l_{0,1}=\l_{n,n+1}=0$, $\s^2_{0,1}=\s^2_{1,2}$, and $\s^2_{n,n+1}=\s^2_{n-1,n}$.
\end{defn}
\vspace{5pt}

One of the simplest examples of a market that satisfies this definition is a {\em first-order model,} and these models can be used to approximate the long-term behavior of more general asymptotically stable markets. First-order models retain certain important characteristics of actual markets. 

\begin{defn} \label{D3.1}  \citep{F:2002, BFK:2005} For $n>1$, a  {\em first-order model} is a family of strictly positive continuous semimartingales $\XX$ defined by
\begin{equation}\label{3.2}
d\log X_i(t) = g_{r_t(i)} \, dt+\s_{r_t(i)} \, dW_i(t),
\end{equation}
where $\s_1,\ldots,\s_n$ are positive constants, $g_1,\ldots,g_n$  are constants satisfying
\begin{equation}\label{3.3}
 g_1+\cdots+g_n = 0 \quad\text{ and }\quad g_1+\cdots+g_k < 0 \text{ for } k < n,
\end{equation}
and $(W_1, \ldots, W_n)$ is a Brownian motion.
\end{defn}  
\vspace{5pt}
 
The first order model $\XX$ defined by \eqref{3.2} is asymptotically stable with 
\begin{equation}\label{4.1}
 \l_{k,k+1} = \limt{1}\L_{k,k+1}(t) =-2\big(g_1+\cdots+g_k\big),\as,
\end{equation}
for $k=1,\ldots,n-1$, and
\begin{equation}\label{4.2}
\s^2_{k,k+1}=      \limt{1}\bbrac{\log X_{(k)}-\log X_{(k+1)}}_t= \s^2_k+\s^2_{k+1},\as,
\end{equation}
for $k=1,\ldots,n-1$ (see Section~3 of \citet{BFK:2005}). 

By Proposition~2.3 of \citet{BFK:2005}, each of the processes $X_i$ in a  first-order-model asymptotically spends equal time in each rank and hence has zero asymptotic log-drift. For a portfolio $\p$ in a market $\XX$ represented by a first-order model, the {\em portfolio growth rate} $\g_\p$ will satisfy
\begin{align}
\g_\p(t)&= \sumi \p_i(t)g_{r_t(i)}+\g^*_\p(t)\notag\\
 &= \sumk \p_{p_t(k)}(t)g_k+\g^*_\p(t),\as,\label{3.4}
\end{align}
where $p_t\in\Sigma_n$ is the inverse permutation to the rank function $r_t$. By Proposition~1.3.1 of \citet{F:2002}, the value process $Z_\p$ of $\p$ satisfies
\[
\limt{1}\Big(Z_\p(t)-\intt\g_\p(s) \, ds\Big)=0,\as
\]

It would seem reasonable that if we swapped the weights of more than one pair of assets, the resulting portfolio would also grow faster than the market. For a market $\XX$, the {\em reverse-weighted portfolio} is the portfolio $\p$ with weight processes
\begin{equation}\label{3.1}
\p_i(t)=\m_{(n+1-r_t(i))}(t),
\end{equation}
for $i=1,\ldots,n$. Note that a reverse-weighted portfolio is distinct from a reciprocal-weighted portfolio in which $\p_i(t) \propto \m^{-1}_i(t)$, for $i=1,\ldots,n$, a case which is discussed in \citet{Vervuurt/Karatzas:2015}.

For $n=2$, the swap portfolio of Proposition~\ref{P1} is the reverse-weighted portfolio, and we have seen that this reverse-weighted portfolio asymptotically outperforms the market. 
%However, when $n\ge4$, complications arise. It turns out that if we swap the weights of more than a single pair of assets, the resulting portfolio is not functionally generated.
Due to Proposition~\ref{P2.1}, for $n > 2$ we cannot represent reverse-weighted portfolios using generating functions. Hence, we introduce the additional structure of a first-order model in order to characterize the long-term behavior of these portfolios.

\begin{prop}\label{P3.1} Suppose that $n>1$ and $\XX$ is  a first-order model \eqref{3.2} for which the excess growth rate of the reverse-weighted portfolio \eqref{3.1} satisfies $\g^*_\p(t)\ge\g^*_\m(t)$,~a.s., for $t\in[0,\infty)$. Then the growth rate of the reverse-weighted portfolio will be greater than that of the market,
\begin{equation}\label{3.5}
 \g_\p(t) > \g_\m(t),\as,
\end{equation}
for $t\in[0,\infty)$, except on a set of Lebesgue measure zero.
\end{prop}
\begin{proof} For the first-order model $\XX$, it follows from \eqref{3.4} that  the market growth rate is
\begin{align*}
\g_\m(t) &= \sumi\m_i(t)g_{r_t(i)} +\g^*_\m(t)\\
&= \sumk \m_{(k)}(t)g_k +\g^*_\m(t),\as
\end{align*}
Similarly, the growth rate for the reverse-weighted portfolio $\p$ is
\begin{align*}
\g_\p(t) &= \sumi\p_i(t)g_{r_t(i)}+\g^*_\p(t)\\
&= \sumk \m_{(n+1-k)}(t)g_k + \g^*_\p(t),\as
\end{align*}
 Hence,
\begin{align*}
\g_\p(t)-\g_\m(t) &= \sumk \big(\m_{(n+1-k)}(t)-\m_{(k)}(t)\big)g_k+\g^*_\p(t)-\g^*_\m(t)\notag\\
 &\ge \sumk \big(\m_{(n+1-k)}(t)-\m_{(k)}(t)\big)g_k\notag\\
&= \sumk \ph_k(t)g_k, \as, %\label{3.6}
\end{align*}
where $\ph_k(t) = \m_{(n+1-k)}(t)-\m_{(k)}(t)$. From the definition of rank it follows that $\ph_{k+1}(t) > \ph_k(t)$,~a.s., for $k=1,\ldots,n-1$ and $t\in[0,\infty)$ except on a set of Lebesgue measure zero. With this inequality and \eqref{3.3}, we can apply summation by parts and obtain 
\begin{align*}
\sumk \ph_k(t)g_k &= \ph_1(t) \sumk g_k + \sum_{k=1}^{n-1} \big(\ph_{k+1}(t)-\ph_k(t)\big) \sum_{\ell=k+1}^n g_\ell \notag\\
&> 0,\as,
\end{align*}
for $t\in[0,\infty)$ except on a set of Lebesgue measure zero, and \eqref{3.5} follows.
\end{proof}

We can apply this to a first-order model with rank-symmetric variance parameters, $\s^2_k=\s^2_{n+1-k}$.

\begin{cor}\label{C3.1} Suppose that $\XX$ is  a first-order model \eqref{3.2} for which $\s^2_k=\s^2_{n+1-k}>0$, for $k=1,\ldots,n$. Then the growth rate of the reverse-weighted portfolio $\p$ will be greater than that of the market,
\begin{equation*}
\g_\p(t)>\g_\m(t),\as,
\end{equation*}
for $t\in[0,\infty)$, except on a set of Lebesgue measure zero.
\end{cor}
\begin{proof} Since both the weights \eqref{3.1}  and the variances  $\s^2_k=\s^2_{n+1-k}$ are reversed by rank, we see from \eqref{1.101} that the excess growth rates  $\g^*_\p(t)=\g^*_\m(t)$, a.s., for $t\in[0,\infty)$.  Hence, the corollary follows from Proposition~\ref{P3.1}.
\end{proof}

First-order models are too restrictive to be used as universal market models, however asymptotically stable markets  can often be approximated by first-order models, at least over the long term. 

\begin{defn}\label{D4.2}\citep{F:2002} 
Let $\{\BX_1,\ldots,\BX_n\}$ be an asymptotically stable  family of strictly positive continuous semimartingales with parameters $\bl_{k,k+1}$ and $\bs^2_{k,k+1}$, for $k=1,\ldots,n$, defined as in 2 and 3 of Definition~\ref{D4.1}. Then the {\em first-order approximation} for $\{\BX_1,\ldots,\BX_n\}$ is the first-order model  $\{X_1,\ldots,X_n\}$ with
\begin{equation}\label{4.4}
d\log X_i(t) = g_{r_t(i)} \, dt+\s_{r_t(i)} \, dW_i(t),
\end{equation}
for $i=1,\ldots,n$, where $r_t\in\Sigma_n$ is the rank function for the $X_i$, the parameters $g_k$ and $\s_k$ are defined by
\begin{equation}\label{4.5}
\begin{split}
g_k  &= \half \bl_{k-1,k} -  \half \bl_{k,k+1},\text{ for } k=1,\ldots,n, \\
\s_k^2 &= \frac{1}{4}\big(\bs^2_{k-1,k}+\bs^2_{k,k+1}\big),\text{ for } k=1,\ldots,n, 
\end{split}
\end{equation}
$\s_k$ is the positive square root of $\s^2_k$, and $(W_1,\ldots,W_n)$ is a Brownian motion.
\end{defn}\vspace{5pt}

For the first-order model \eqref{4.4} with parameters \eqref{4.5}, equations \eqref{4.1} and \eqref{4.2} imply that
\begin{equation*} %\label{4.6}
 \l_{k,k+1}=-2\big(g_1+\cdots+g_k\big)= \bl_{k,k+1},\as,
\end{equation*}
for $k=1,\ldots,n-1$, and
\begin{equation*} %\label{4.7}
\s^2_{k,k+1}=\s^2_k+\s^2_{k+1}
= \frac{1}{4}\big(\bs^2_{k-1,k}+2\bs^2_{k,k+1}+\bs^2_{k+1,k+2}\big),\as,
\end{equation*}
for $k=1,\ldots,n-1$. 

If the behavior of the first-order approximation $\XX$ is close enough to that of the original market $\BXX$, then we can draw conclusions about portfolio behavior in $\BXX$ from the behavior of the corresponding portfolio in $\XX$, at least if the portfolio is {\em na\"ive} in the sense that sophisticated asset selection is not involved (this idea is developed in \citet{Arnott:2013} and  \citet{JIC:2019}). If the first-order approximation behaves enough like $\BXX$,  and if either the hypotheses of   Proposition~\ref{P3.1} or those of  Corollary~\ref{C3.1} hold, at least approximately, then the reverse-weighted portfolio will probably have a higher growth rate than the market portfolio. In the next section we apply these methods to a set of empirical data.

\section{Reverse-weighted portfolios of commodity futures} \label{data}%%%%%%%%%%%%%%%%%%%

We examine the applicability of our theoretical results using monthly historical futures prices data for 26 commodities from 1977-2018. Table~\ref{commInfoTab} lists the start month and trading market for the 26 commodity futures contracts in our data set. We construct equal-weighted, price-weighted, and reverse price-weighted portfolios of commodity futures and examine their performance over this time period. Following \citet{Vervuurt/Karatzas:2015}, we also construct and examine the performance of the diversity-weighted portfolio of commodity futures with a negative diversity parameter of -0.5. For this application to commodity futures, the price-weighted portfolio corresponds to the market portfolio $\m$ defined in Section~\ref{markets}, and the reverse price-weighted portfolio corresponds to the reverse-weighted portfolio~\eqref{3.1} defined in Section~\ref{fom}. 

\subsection*{Implied commodity futures prices}
The most liquid commodity futures contracts are usually those with expiration dates approximately one or two months in the future. Accordingly, we use two-month commodity futures contracts to generate our data set whenever such prices exist. We define implied two-month futures prices to fill the gaps when the available data do not include actual two-month futures prices. From these implied prices we can define time series that can be approximated with a first-order model following the methodology of the previous section. 

Suppose that commodity $i$ has a futures contract with expiration date $\t \in \N$. For $t \in \N$, with $t \le \t$, let
\begin{equation*} 
 F_i(t, \t) \; = \; \text{futures price for commodity $i$ at time $t$ for the contract with expiration at} \; \t.
\end{equation*} 
In this case, $F_i(\t,\t)$ is the {\em spot price} for commodity $i$ at time $\t$. For $t \in \N$, let $\nu_2 > \nu_1 \ge 0$, $\nu_1,\nu_2 \ne 2$, be the smallest integers closest to two such that $t + \nu_1$ and $t + \nu_2$ are expiration dates of futures contracts for the $i$th commodity. We define the \emph{carry factor} for commodity $i$ at time $t$ to be
\begin{equation*} %\label{carry-rate}
 \D_i(t) \eqdef \frac{\log F_i(t, t + \nu_2) - \log F_i(t, t + \nu_1)}{\nu_2 - \nu_1}.
\end{equation*}
Note that the carry factor $\D_i$ is commodity-specific and can vary over time, and it is calculated using futures contracts with expirations as close as possible, but not equal to, two months in the future. 

\begin{defn}\label{impliedFuture}
For $t \in \N$, let $\nu \ge 0$, be the smallest of the closest integers to two such that $t + \nu$ is an expiration date of a futures contract for the $i$th commodity. Then the {\em implied two-month futures price} at time $t$ for commodity $i$ is  
\begin{equation*} %\label{impliedFutureEq}
 \FF_i(t, t+2) \eqdef  e^{(2-\nu)\D_i(t)}F_i(t, t+\nu).
 \end{equation*} 
\end{defn}

We generate monthly time series $\BX_1(t),\ldots,\BX_n(t)$, for $t\in\N$, from the implied two-month futures prices of each of the $n=26$ commodities in our data set. We use these data to compare and rank the commodities over the 1977-2018 time period. In terms of Definition~\ref{impliedFuture}, we let
\begin{equation}\label{5.2}
 \BX_i(t) \eqdef  \FF_i(t, t+2),
\end{equation}
for $t=1,2,\ldots, 492$, and $i=1, \ldots, 26$. It follows from Definition~\ref{impliedFuture} and \eqref{5.2} that if $F_i(t,t+2)$ exists, then
\begin{equation*}
 \BX_i(t) = \FF_i(t, t + 2) = F_i(t, t + 2),
\end{equation*}
so that in this case the two-month futures price is equal to the two-month implied futures price.

In order to meaningfully compare and rank the implied two-month futures prices of the 26 different commodities in our data set, it is necessary to normalize these prices. We set the initial implied futures prices of all commodities with available futures contracts on the November 1968 data start month --- soybean meal, soybean oil, and soybeans --- equal to each other. All subsequent monthly price changes occur without modification, meaning that implied futures price dynamics are unaffected by our normalization. This method of normalizing prices is similar to \citet{Asness/Moskowitz/Pedersen:2013}, who rank commodity futures based on each commodity's current spot price relative to its average spot price 4.5 to 5.5 years in the past.

For those commodities that enter into our data set after November 1968, we set the initial implied futures log-price equal to the average log-price of those commodities already in our data set on that month. After a commodity enters into the data set with a normalized price, all subsequent monthly price changes occur without modification. The resulting normalized implied two-month futures prices for the 26 commodities in our data set are plotted in Figure~\ref{pricesFig}, with the log-prices reported relative to the average for all commodities in each month.

Figure \ref{pricesFig} suggests that implied two-month commodity futures prices are asymptotically stable. Although there is no formal test of asymptotic stability, the figure shows that no commodity drops out of the market during the 1968-2018 time period, consistent with part 1 of Definition \ref{D4.1} (coherence). Furthermore, the relative implied prices of differently ranked commodities appear approximately constant in Figure \ref{pricesFig}, which is consistent with parts 2 and 3 of Definition \ref{D4.1}. We confirm that the first-order approximation of the two-month implied commodity futures market provides an accurate description of this market below.

\subsection*{Commodity futures returns and market efficiency}
When forming equal-weighted, price-weighted, diversity-weighted, and reverse-weighted portfolios, we hold two-month futures contracts in all months in which such contracts exist. For those months in which there are no two-month futures contracts, we hold those contracts with the next expiration horizon greater than two months in the future. In both cases, the change in the implied two-month futures price of commodity $i$, $d \log \BX_i(t)$, is not necessarily equal to the return from holding the underlying commodity futures contract, $d \log F_i(t, \t)$, where $\t \geq t + 2$ is a futures contract expiration date. We refer to the difference between the log change in the implied two-month futures price and the return from holding the underlying futures contract as the {\em carry}. This carry satisfies
\begin{equation} \label{carryEq}
 C_i(t) \, dt = d \log F_i(t, \t) - d \log \BX_i(t),
\end{equation}
for $i = 1, \ldots, n$ and $t \in [0, \infty)$. 

%Alternatively, \eqref{carryEq} can be written as
%\begin{equation} \label{carryEq}
% d \log \BX_i(t) + C_i(t) \, dt = d \log F_i(t, \t),
%\end{equation}
%for $i = 1, ldots, n$ and $t \in [0, \infty)$. 

The carry \eqref{carryEq} measures the gap between the returns from holding commodity futures contracts and changes in the implied futures prices. For a portfolio, the carry is just the weighted sum of the carries for each commodity futures contract, with the weights corresponding to the portfolio weights. Since the price-weighted, reverse-weighted, equal-weighted, and diversity-weighted portfolios weight each futures contract differently, the carry for each portfolio will also likely be different. In order for the commodity futures market to be efficient, the different carries for the portfolios must balance the different growth rates of the implied two-month futures prices in those portfolios in such a way that all portfolio returns are approximately equal to each other. 

More precisely, the first-order approximation of the family of two-month implied futures prices $\{\BX_1, \ldots, \BX_n\}$ features differential growth rates $g_k$ and variances $\s^2_k$. If these parameters satisfy the conditions of Proposition~\ref{P3.1} or Corollary~\ref{C3.1}, then the growth rate of the implied futures prices in the reverse-weighted portfolio $\g_{\p}$ will likely exceed that of the implied futures prices in the price-weighted portfolio $\g_{\m}$. Therefore, in order for these two portfolios to have equal log returns --- as expected for an efficient market --- it must be that the carry for the reverse-weighted portfolio is consistently below that for the price-weighted portfolio in a way that approximately balances the differential growth rates $\g_{\p}$ and $\g_{\m}$.

\subsection*{Portfolios of commodity futures}
We analyze and compare the performance of price-weighted, equal-weighted, diversity-weighted, and reverse price-weighted portfolios of commodity futures. The price-weighted portfolio is simply the market portfolio $\m$ defined in Section~\ref{markets}, with weight processes 
\[
\m_i(t) = \frac{\BX_i(t)}{\BX_1(t)+\cdots+\BX_n(t)},
\]
for $i = 1, \ldots, n$, where $\BX_1, \ldots, \BX_n$ are the two-month implied futures prices defined above. Because futures do not have a size, the market portfolio $\m$ weights each futures contract by its implied price $\BX_i$. The equal-weighted portfolio is the portfolio with weights constant and equal to $1/n$ for all $t$. The diversity-weighted portfolio is the portfolio with weights
\begin{equation*} %\label{diversityWeightsEq}
 \p_i(t) = \frac{\BX^p_i(t)}{\BX^p_1(t)+\cdots+\BX^p_n(t)},
\end{equation*} 
for $i = 1, \ldots, n$, with the diversity parameter $p$ set equal to $-0.5$ following \citet{Vervuurt/Karatzas:2015}. Finally, the reverse price-weighted portfolio is the reverse-weighted portfolio defined in Section~\ref{fom}, with $\p_i(t) = \m_{(n + 1 - r_t(i))}$, for $i = 1, \ldots, n$.

Although our commodity futures data cover 1968-2018, the fact that we normalize implied futures prices by setting them equal to each other on the November 1968 start date, as discussed above, implies that these prices cannot be meaningfully compared and ranked until they have time to disperse. Thus, for each commodity in our data set, we wait five years after the start of its price data before including that commodity futures contract in our equal-weighted, price-weighted, diversity-weighted, and reverse-weighted portfolios. Furthermore, we do not start forming portfolios until we have at least ten commodities with at least five years of price data, which occurs in November 1977. As a consequence, all results we report for the different portfolios run from November 1977 to January 2018.

In Table~\ref{returnsTab}, we report the average and standard deviation of the annual log-returns for equal-weighted, price-weighted, diversity-weighted, and reverse-weighted portfolios of commodity futures from 1977-2018. The cumulative returns for these three portfolios are shown in Figure~\ref{returnsFig}. This figure clearly shows that the reverse-weighted portfolio grows faster than the price-weighted market portfolio over the 1977-2018 time period, consistent with the result in Proposition~\ref{P3.1}. The reverse-weighted portfolio also grows faster than the equal-weighted and diversity-weighted portfolios. The table also shows that the equal-weighted, diversity-weighted, and reverse-weighted portfolios all have lower returns standard deviations than the price-weighted market portfolio, despite the fact that these three portfolios' average returns are higher.

Table~\ref{returnsTab} also reports the Sharpe ratios of the equal-weighted, diversity-weighted, and reverse-weighted portfolios, defined as the average log-returns of each of the three portfolios minus the log-returns of the price-weighted portfolio divided by the standard deviation of these relative returns. The cumulative relative returns for the equal-weighted, diversity-weighted, and reverse-weighted portfolios are shown in Figure~\ref{relReturnsFig}. The results in the table and figure confirm the faster growth of the reverse-weighted portfolio relative to the equal-weighted and diversity-weighted portfolios, and also reveal a higher Sharpe ratio for the reverse-weighted portfolio.

According to \eqref{1.1}, the log-returns of each of the the three commodity futures portfolios shown in Figure~\ref{returnsFig} can be decomposed into the weighted growth rate of the individual futures contracts and the excess growth rate process $\g^*$, which is given by \eqref{1.2}. Figure~\ref{gammaStarFig} plots the cumulative values of the excess growth rate process for the price-weighted, reverse-weighted, equal-weighted, and diversity-weighted portfolios of commodity futures from 1977-2018. These excess growth rates are calculated using the decomposition \eqref{1.1} together with the log-returns of each portfolio and the weighted log-returns of the individual futures contracts held in each portfolio. According to the figure, the processes $\g^*$ for the reverse-weighted and price-weighted portfolios are approximately equal to each other, which is an important condition needed to apply Proposition~\ref{P3.1} regarding the growth rates of the two portfolios. Therefore, Figure~\ref{gammaStarFig} together with Proposition~\ref{P3.1} partially explains the outperformance of the reverse-weighted portfolio relative to the price-weighted portfolio shown in Figure~\ref{returnsFig}.

As discussed earlier, the returns of the price-weighted, reverse-weighted, equal-weighted, and diversity-weighted portfolios are not equal only to changes in the implied futures prices in these portfolios but also to the differential carry for each portfolio. According to Figure~\ref{gammaStarFig} together with Proposition~\ref{P3.1}, market efficiency requires that the reverse-weighted portfolio have a consistently lower carry than the price-weighted portfolio, so that this lower carry can cancel out the higher growth rate of implied futures prices in the reverse-weighted portfolio $\g_{\p}$ implied by the proposition.

Figure~\ref{carryFig} plots the cumulative carry for each of these four portfolios of commodity futures. According to the figure, the carry for each of the four portfolios is consistently negative, with the most negative carry for the reverse-weighted portfolio and the least negative carry for the price-weighted portfolio. Furthermore, the magnitude of the cumulative effects of this carry on returns is meaningfully large. Nonetheless, Figure~\ref{returnsFig} shows that the carry for the reverse-weighted portfolio is not far enough below the carry for the price-weighted portfolio so as to cancel out the higher growth rate of implied futures prices in the reverse-weighted portfolio. According to this figure, the differential carry of the reverse-weighted and price-weighted portfolios shown in Figure~\ref{carryFig} is not sufficient to equate the log returns of the two portfolio. Thus, there remains an inefficiency in the commodity futures market.

%As discussed earlier, returns from holding commodity futures contracts are not equal to changes in implied futures prices, with the difference being the carry as defined by \eqref{carryEq}. Because the price-weighted, reverse-weighted, equal-weighted, and diversity-weighted portfolios weight each futures contract differently, it follows that the carry for each portfolio will also be different. Figure~\ref{carryFig} plots the cumulative carry for these four portfolios of commodity futures. According to the figure, the carry for all four portfolios is consistently negative, with the most negative carry for the reverse-weighted portfolio and the least negative carry for the price-weighted portfolio. Furthermore, the magnitude of the cumulative effects of this carry on returns is meaningfully large --- the reverse-weighted portfolio outperforms the price-weighted, equal-weighted, and diversity-weighted portfolios despite a substantially more negative carry.

\subsection*{First-order approximation of implied futures market}
In order to understand the outperformance of the reverse-weighted portfolio relative to the price-weighted portfolio as shown in Table~\ref{returnsTab} and Figure~\ref{returnsFig}, we estimate the first-order approximation of the two-month normalized implied futures market plotted in Figure~\ref{pricesFig}. According to Definition~\ref{D3.1}, a first-order model is defined only for a fixed number of ranked assets. Therefore,  we estimate the first-order approximation of the implied futures market using commodity futures price data starting on April 1995, since this date is five years after the last commodity futures contract price data begin (Table~\ref{commInfoTab}) and hence it is also the last date on which a new commodity enters our data set. The total number of commodity implied futures prices is thus fixed at 26 from April 1995 to January 2018, and we are able to estimate the first-order approximation of this market over that time period.

Figure~\ref{GsFig} plots the parameters $g_k$, defined by \eqref{4.5}, for the first-order approximation of the two-month implied futures market. This figure plots the values of these parameters after applying a reflected Gaussian filter with a bandwidth of six ranks together with the unfiltered values, which are represented by the red circles in the figure. Figure~\ref{GsFig} shows that the first-order approximation of this market features mostly higher growth rates $g_k$ for lower-ranked commodity futures than for higher-ranked futures, and that the sum of these growth rates is negative for top-ranked subsets. This pattern is consistent with the stability condition \eqref{3.3} for the $g_k$ parameters of a first-order model.

Figure~\ref{sigmasFig} plots both filtered and unfiltered values of the parameters $\s_k$, defined by \eqref{4.5}, for this same first-order approximation. This figure shows that the filtered values of the volatility parameters $\s_k$ are roughly constant across ranks, which implies that the conditions of Corollary~\ref{C3.1} are approximately satisfied by the commodity implied futures market. According to Corollary~\ref{C3.1}, then, the growth rate of the reverse-weighted portfolio of commodity futures should be greater than that of the price-weighted portfolio, before adjusting for the carry of each portfolio. In fact, we find that the growth rate of reverse-weighted implied commodity futures is enough larger than that of price-weighted implied commodity futures that the reverse-weighted portfolio substantially outperforms the price-weighted portfolio, despite its significantly more negative carry as shown in Figure~\ref{carryFig}.

%Taken together, Figure~\ref{gammaStarFig} and Figure~\ref{sigmasFig} imply that the conditions of Corollary~\ref{C3.1} are approximately satisfied by the commodity implied futures market, since the former figure shows that the excess growth rates $\g^*$ for the reverse-weighted and price-weighted portfolios are roughly equal. 

A closer inspection of Figure~\ref{sigmasFig} reveals that the unfiltered estimates of the parameters $\s_k$ are highest and approximately equal to each other at the top two and bottom two ranks. Furthermore, these volatility parameters form a roughly symmetric U-shape when plotted versus rank. The results of Corollary~\ref{C3.1} still apply to such a first-order model, since $\s^2_k = \s^2_{n + 1 - k} > 0$ for all $k = 1, \ldots, n$ in the case of a symmetric U-shape for the volatility parameters. Thus, both the filtered and unfiltered estimates of the volatility parameters $\s_k$ are approximately consistent with the outperformance of the reverse-weighted portfolio shown in Figure~\ref{returnsFig}, at least in the absence of a different carry for the two portfolios.

Finally, Figure~\ref{simVsDataFig} presents a log-log plot of average relative implied futures prices for different ranked prices versus rank for 1995-2018 together with the average relative prices from 10,000 simulations of the first-order model \eqref{4.4} using the estimated parameters $g_k$ and $\s_k$ from Figures~\ref{GsFig} and~\ref{sigmasFig}. This figure shows that the simulated first-order approximation provides a reasonably accurate match to the actual relative commodity implied futures price distribution observed over this time period.

%It is straightforward to extend the result of Corollary~\ref{C3.1} to a first-order model in which the volatility parameters $\s_k$ follow such a U-shaped pattern. Any first-order model featuring $\s_k$ that are reverse-rank symmetric, so that $\s_k = \s_{n + 1 - k}$ for all $k = 1, \ldots, n$, will generate equal excess growth rate processes $\g^*$ for the reverse-weighted and price-weighted portfolios, since reversing the weights for such a model will have no effect on the sum \eqref{1.2}. 

%Finally, Figure~\ref{predVsDataFig} presents a log-log plot of average relative implied futures prices for different ranked prices versus rank for 1995-2018 together with the predicted relative prices according to \eqref{4.8} using the estimated parameters $g_k$ and $\s_k$ from Figures~\ref{GsFig} and~\ref{sigmasFig}. This figure shows that the predicted relative prices according to the first-order approximation using the estimated parameters $g_k$ and $\s_k$ provides a reasonably accurate match to the actual average relative commodity implied futures prices observed over this time period. 

%\pagebreak
\bibliographystyle{chicago}
\bibliography{math2,math3}

\begin{thebibliography}{}

\bibitem[\protect\citeauthoryear{Arnott, Hsu, Kalesnik, and Tindall}{Arnott
  et~al.}{2013}]{Arnott:2013}
Arnott, R., J.~Hsu, V.~Kalesnik, and P.~Tindall (Summer 2013).
\newblock The surprising alpha from {Malkiel's} monkey and upside-down
  strategies.
\newblock {\em Journal of Portfolio Management\/}~{\em 39\/}(4), 1--16.

\bibitem[\protect\citeauthoryear{Asness, Moskowitz, and Pedersen}{Asness
  et~al.}{2013}]{Asness/Moskowitz/Pedersen:2013}
Asness, C.~S., T.~J. Moskowitz, and L.~H. Pedersen (2013).
\newblock Value and momentum everywhere.
\newblock {\em Journal of Finance\/}~{\em 68\/}(3), 929--985.

\bibitem[\protect\citeauthoryear{Banner, Fernholz, and Karatzas}{Banner
  et~al.}{2005}]{BFK:2005}
Banner, A., R.~Fernholz, and I.~Karatzas (2005).
\newblock On {Atlas} models of equity markets.
\newblock {\em Annals of Applied Probability\/}~{\em 15}, 2296--2330.

\bibitem[\protect\citeauthoryear{Banner, Fernholz, Papathanakos, Ruf, and
  Schofield}{Banner et~al.}{2019}]{JIC:2019}
Banner, A., R.~Fernholz, V.~Papathanakos, J.~Ruf, and D.~Schofield (2019).
\newblock Diversification, volatility, and surprising alpha.
\newblock {\em Journal of Investment Consulting\/}~{\em 19\/}(1), 23--30.

\bibitem[\protect\citeauthoryear{Banner and Ghomrasni}{Banner and
  Ghomrasni}{2008}]{Banner/Ghomrasni:2008}
Banner, A. and R.~Ghomrasni (2008, July).
\newblock Local times of ranked continuous semimartingales.
\newblock {\em Stochastic Processes and their Applications\/}~{\em 118\/}(7),
  1244--1253.

\bibitem[\protect\citeauthoryear{Bass and Pardoux}{Bass and
  Pardoux}{1987}]{Bass/Pardoux:1987}
Bass, R. and E.~Pardoux (1987).
\newblock Uniqueness for diffusions with piecewise constant coefficients.
\newblock {\em Probability Theory and Related Fields\/}~{\em 76}, 557--572.

\bibitem[\protect\citeauthoryear{Bruggeman}{Bruggeman}{2016}]{Bruggeman:2016}
Bruggeman, C. (2016).
\newblock {\em Dynamics of Large Rank-Based Systems of Interacting Diffusions}.
\newblock Ph.\ D. thesis, Columbia University.

\bibitem[\protect\citeauthoryear{Chatterjee and Pal}{Chatterjee and
  Pal}{2010}]{Chatterjee/Pal:2013}
Chatterjee, S. and S.~Pal (2010).
\newblock A phase transition behavior for {Brownian} motions interacting
  through their ranks.
\newblock {\em Probability Theory and Related Fields\/}~{\em 147\/}(1--2),
  123--159.

\bibitem[\protect\citeauthoryear{Dembo, Jara, and Olla}{Dembo
  et~al.}{2017}]{Dembo/Jara/Olla:2017}
Dembo, A., M.~Jara, and S.~Olla (2017).
\newblock The infinite {Atlas} process: Convergence to equilibrium.
\newblock {\em Ann. Inst. H. Poincar\'e Probab. Statist.\/}~{\em 55\/}(2),
  607--619.

\bibitem[\protect\citeauthoryear{Dembo, Shkolnikov, Varadhan, and
  Zeitouni}{Dembo et~al.}{2016}]{DSVZ:2016}
Dembo, A., M.~Shkolnikov, {\relax S.\ R.\ S}.~Varadhan, and O.~Zeitouni (2016).
\newblock Large deviations for diffusions interacting through their ranks.
\newblock {\em Comm. Pure Appl. Math.\/}~{\em 69\/}(7), 1259--1313.

\bibitem[\protect\citeauthoryear{Dembo and Tsai}{Dembo and
  Tsai}{2017}]{Dembo/Tsai:2017}
Dembo, A. and L.-C. Tsai (2017).
\newblock Equilibrium fluctuation of the {Atlas} model.
\newblock {\em Annals of Probability\/}~{\em 45\/}(6B), 4529--4560.

\bibitem[\protect\citeauthoryear{Fernholz}{Fernholz}{2002}]{F:2002}
Fernholz, E.~R. (2002).
\newblock {\em Stochastic Portfolio Theory}.
\newblock New York: Springer-Verlag.

\bibitem[\protect\citeauthoryear{Fernholz}{Fernholz}{1999}]{F:pgf}
Fernholz, R. (1999).
\newblock Portfolio generating functions.
\newblock In M.~Avellaneda (Ed.), {\em {Quantitative Analysis in Financial
  Markets}}, River Edge, NJ. World Scientific.

\bibitem[\protect\citeauthoryear{Fernholz}{Fernholz}{2001}]{F:rank}
Fernholz, R. (2001).
\newblock Equity portfolios generated by functions of ranked market weights.
\newblock {\em Finance and Stochastics\/}~{\em 5}, 469--486.

\bibitem[\protect\citeauthoryear{Fernholz}{Fernholz}{2016}]{F:2016}
Fernholz, R. (2016).
\newblock A new decomposition of portfolio return.
\newblock {\em ArXiv:1606.05877\/}, 1--4.

\bibitem[\protect\citeauthoryear{Fernholz}{Fernholz}{2017a}]{F:2017}
Fernholz, R. (2017a).
\newblock Stratonovich representation of semimartingale rank processes.
\newblock {\em ArXiv:1705.00336\/}.

\bibitem[\protect\citeauthoryear{Fernholz, Ichiba, and Karatzas}{Fernholz
  et~al.}{2013a}]{FIK:2013a}
Fernholz, R., T.~Ichiba, and I.~Karatzas (2013a).
\newblock A second-order stock market model.
\newblock {\em Annals of Finance\/}~{\em 9}, 1--16.

\bibitem[\protect\citeauthoryear{Fernholz, Ichiba, and Karatzas}{Fernholz
  et~al.}{2013b}]{FIK:2013b}
Fernholz, R., T.~Ichiba, and I.~Karatzas (2013b).
\newblock Two {Brownian} particles with rank-based characteristics and
  skew-elastic collisions.
\newblock {\em Stochastic Processes and their Applications\/}~{\em 123},
  2999--3026.

\bibitem[\protect\citeauthoryear{Fernholz, Ichiba, Karatzas, and
  Prokaj}{Fernholz et~al.}{2013}]{FIKP:2013}
Fernholz, R., T.~Ichiba, I.~Karatzas, and V.~Prokaj (2013).
\newblock A planar diffusion with rank-based characteristics and perturbed
  {Tanaka} equations.
\newblock {\em Probability Theory and Related Fields\/}~{\em 156}, 343--374.

\bibitem[\protect\citeauthoryear{Fernholz and Karatzas}{Fernholz and
  Karatzas}{2009}]{FK:2009}
Fernholz, R. and I.~Karatzas (2009).
\newblock Stochastic portfolio theory: an overview.
\newblock In A.~Bensoussan and Q.~Zhang (Eds.), {\em Mathematical Modelling and
  Numerical Methods in Finance: Special Volume, Handbook of Numerical
  Analysis}, Volume~XV, pp.\  89--168. Amsterdam: North-Holland.

\bibitem[\protect\citeauthoryear{Fernholz, Karatzas, and Kardaras}{Fernholz
  et~al.}{2005}]{FKK:2005}
Fernholz, R., I.~Karatzas, and C.~Kardaras (2005).
\newblock Diversity and relative arbitrage in financial markets.
\newblock {\em Finance and Stochastics\/}~{\em 9}, 1--27.

\bibitem[\protect\citeauthoryear{Fernholz, Karatzas, and Ruf}{Fernholz
  et~al.}{2018}]{FKR:2018}
Fernholz, R., I.~Karatzas, and J.~Ruf (2018).
\newblock Volatility and arbitrage.
\newblock {\em Annals of Applied Probability\/}~{\em 28\/}(1), 378--417.

\bibitem[\protect\citeauthoryear{Fernholz}{Fernholz}{2017b}]{Fernholz:2017a}
Fernholz, R.~T. (2017b).
\newblock Nonparametric methods and local-time-based estimation for dynamic
  power law distributions.
\newblock {\em Journal of Applied Econometrics\/}~{\em 32\/}(7), 1244--1260.

\bibitem[\protect\citeauthoryear{Ichiba and Karatzas}{Ichiba and
  Karatzas}{2010}]{IK:2010}
Ichiba, T. and I.~Karatzas (2010).
\newblock On collisions of {Brownian} particles.
\newblock {\em Annals of Applied Probability\/}~{\em 20}, 951--977.

\bibitem[\protect\citeauthoryear{Ichiba, Karatzas, and Shkolnikov}{Ichiba
  et~al.}{2013}]{IchibaKaratzasShkolnikov:2013}
Ichiba, T., I.~Karatzas, and M.~Shkolnikov (2013).
\newblock Strong solutions of stochastic equations with rank-based
  coefficients.
\newblock {\em Probability Theory and Related Fields\/}~{\em 156\/}((1-2)),
  229--248.

\bibitem[\protect\citeauthoryear{Ichiba, Pal, and Shkolnikov}{Ichiba
  et~al.}{2013}]{IchibaPalShkolnikov-2013}
Ichiba, T., S.~Pal, and M.~Shkolnikov (2013).
\newblock Convergence rates for rank-based models with applications to
  portfolio theory.
\newblock {\em Probability Theory and Related Fields\/}~{\em 156\/}(1--2),
  415--448.

\bibitem[\protect\citeauthoryear{Ichiba, Papathanakos, Banner, Karatzas, and
  Fernholz}{Ichiba et~al.}{2011}]{IPBKF:2011}
Ichiba, T., V.~Papathanakos, A.~Banner, I.~Karatzas, and R.~Fernholz (2011).
\newblock Hybrid {Atlas} models.
\newblock {\em Annals of Applied Probability\/}~{\em 21}, 609--644.

\bibitem[\protect\citeauthoryear{Karatzas and Kim}{Karatzas and
  Kim}{2020}]{KK:2020}
Karatzas, I. and D.~Kim (2020).
\newblock Trading strategies generated pathwise by functions of market weights.
\newblock {\em Finance and Stochastics\/}~{\em 24\/}(2), 423--463.

\bibitem[\protect\citeauthoryear{Karatzas and Shreve}{Karatzas and
  Shreve}{1991}]{Karatzas/Shreve:1991}
Karatzas, I. and S.~E. Shreve (1991).
\newblock {\em {Brownian Motion and Stochastic Calculus}}.
\newblock New York, NY: Springer-Verlag.

\bibitem[\protect\citeauthoryear{Kartzas and Ruf}{Kartzas and
  Ruf}{2017}]{KR:2017}
Kartzas, I. and J.~Ruf (2017).
\newblock Trading strategies generated by {Lyapunov} functions.
\newblock {\em Finance and Stochastics\/}~{\em 21}, 753--787.

\bibitem[\protect\citeauthoryear{{Khas'minskii}}{{Khas'minskii}}{1960}]{Khasminskii:1960}
{Khas'minskii}, R.~Z. (1960).
\newblock Ergodic properties of recurrent diffusion processes, and
  stabilization of the solution to the cauchy problem for parabolic equations.
\newblock {\em Theory Probab. Appl.\/}~{\em 5}, 179--196.

\bibitem[\protect\citeauthoryear{{Khas'minskii}}{{Khas'minskii}}{1980}]{Khasminskii:1980}
{Khas'minskii}, R.~Z. (1980).
\newblock {\em Stochastic Stability of Differential Equations}.
\newblock Amsterdam: Sijthoff and Noordhoff.

\bibitem[\protect\citeauthoryear{Pal and Pitman}{Pal and
  Pitman}{2008}]{Pal/Pitman:2008}
Pal, S. and J.~Pitman (2008).
\newblock One-dimensional brownian particle systems with rank-dependent drifts.
\newblock {\em Annals of Applied Probability\/}~{\em 18\/}(6), 2179--2207.

\bibitem[\protect\citeauthoryear{Protter}{Protter}{1990}]{protter:1990}
Protter, P. (1990).
\newblock {\em Stochastic Integrals and Differential Equations}.
\newblock Berlin, Heidelberg: Springer-Verlag.

\bibitem[\protect\citeauthoryear{Russo and Vallois}{Russo and
  Vallois}{2006}]{Russo:2006}
Russo, F. and P.~Vallois (2006).
\newblock Elements of stochastic calculus via regularisation.
\newblock {\em ArXiv:0603224\/}, 1--39.

\bibitem[\protect\citeauthoryear{Sarantsev}{Sarantsev}{2015}]{Sarantsev:2015}
Sarantsev, A. (2015).
\newblock Triple and simultaneous collisions of competing {Brownian} particles.
\newblock {\em Electron. J. Probab.\/}~{\em 20\/}(29), 1--28.

\bibitem[\protect\citeauthoryear{Stanley}{Stanley}{2012}]{Stanley:2012}
Stanley, R.~P. (2012).
\newblock {\em Enumerative Combinatorics\/} (Second ed.), Volume~1.
\newblock Cambridge, U.K.: Cambridge University Press.

\bibitem[\protect\citeauthoryear{Strong}{Strong}{2014}]{strong:2014}
Strong, W. (2014).
\newblock Generalizations of functionally generated portfolios with
  applications to statistical arbitrage.
\newblock {\em SIAM J. Financial Math.\/}~{\em 5}, 472--492.

\bibitem[\protect\citeauthoryear{Stroock and Varadhan}{Stroock and
  Varadhan}{2006}]{Stroock/Varadhan:2006}
Stroock, D.~W. and {\relax S.\ R.\ S}.~Varadhan (2006).
\newblock {\em Multidimensional Diffusion Processes}.
\newblock Berlin: Springer.

\bibitem[\protect\citeauthoryear{Vervuurt and Karatzas}{Vervuurt and
  Karatzas}{2015}]{Vervuurt/Karatzas:2015}
Vervuurt, A. and I.~Karatzas (2015).
\newblock Diversity-weighted portfolios with negative parameter.
\newblock {\em Annals of Finance\/}~{\em 11\/}(3--4), 411--432.

\end{thebibliography}

\pagebreak

\begin{table}[H]
\vspace{-15pt}
\begin{center}
\setlength{\extrarowheight}{3pt}
\begin{tabular} {| l ||c|c|}

\hline

      Commodity  & Exchange           & Start            \\
                          &   where Traded   & Date             \\

\hline

  Soybean Meal      &  CBOT        &  11/1968   \\
  Soybean Oil          &  CBOT       &  11/1968   \\
  Soybeans             &  CBOT        &  11/1968    \\
  Wheat                   &  CBOT        &  1/1969    \\
  Corn                     &  CBOT        &  1/1969    \\
  Live Hogs             &  CME          &  12/1969   \\
  Cotton                  &  NYBOT      &  10/1972   \\
  Silver                   &  COMEX      &  10/1972   \\
  Orange Juice       &  CEC           &  11/1972    \\
  Platinum              &  NYMEX      &  11/1972    \\
  Sugar                   &  CSC           &  1/1973    \\
  Lumber                &  CME           &  7/1973    \\
  Coffee                  &  CSC           &  10/1973    \\
  Oats                     &  CBOT        &  10/1974   \\
  Gold                     &  COMEX     &  1/1975   \\
  Live Cattle            &  CME          &  4/1976   \\
  Wheat, K.C.         &  KCBT         &  5/1976    \\
  Feeder Cattle      &  CME           &  11/1977    \\
  Heating Oil          &  NYMEX      &  10/1979    \\
  Cocoa                 &  CSC            &  1/1981    \\
  Wheat, Minn.      &  MGE           &  1/1981    \\
  Palladium            &  NYMEX      &  1/1983    \\
  Crude Oil            &  NYMEX       &  4/1983    \\
  Rough Rice         &  CBOT         &  9/1986    \\
  Copper                &  COMEX     &  11/1988    \\
  Natural Gas         &  NYMEX     &  4/1990    \\
  
\hline

\end{tabular}
\end{center}
\vspace{-5pt} \caption{List of commodity futures contracts along with the exchange where each commodity is traded and the date each commodity started trading.}
\label{commInfoTab}
\end{table}

\begin{table}[H]
\vspace{10pt}
\begin{center}
\setlength{\extrarowheight}{3pt}
\begin{tabular} {|c||c|c|c|c|}

\hline

     & Price-Weighted    & Equal-Weighted & Diversity-Weighted & Reverse Price-Weighted  \\
     & Portfolio               & Portfolio             & Portfolio                   & Portfolio                            \\

\hline

  Average                         &                 -1.43\%       &                  0.43\%        &   1.09\%     &    1.83\%          \\
  Standard Deviation        &                15.38\%       &                 13.79\%       &   13.68\%   &    13.85\%        \\
  Sharpe Ratio                 &                                     &                  0.40            &   0.41         &     0.47              \\
 
\hline

\end{tabular}
\end{center}
\vspace{-5pt} \caption{Annual average and standard deviation of log-returns for price-weighted, equal-weighted, diversity-weighted, and reverse price-weighted portfolios, and Sharpe ratio of equal-weighted, diversity-weighted, and reverse price-weighted portfolios relative to the price-weighted portfolio, 1977-2018.}
\label{returnsTab}
\end{table}

\pagebreak

\begin{figure}[H]
\begin{center}
\vspace{-15pt}
\hspace{-20pt}\scalebox{.67}{ \includegraphics{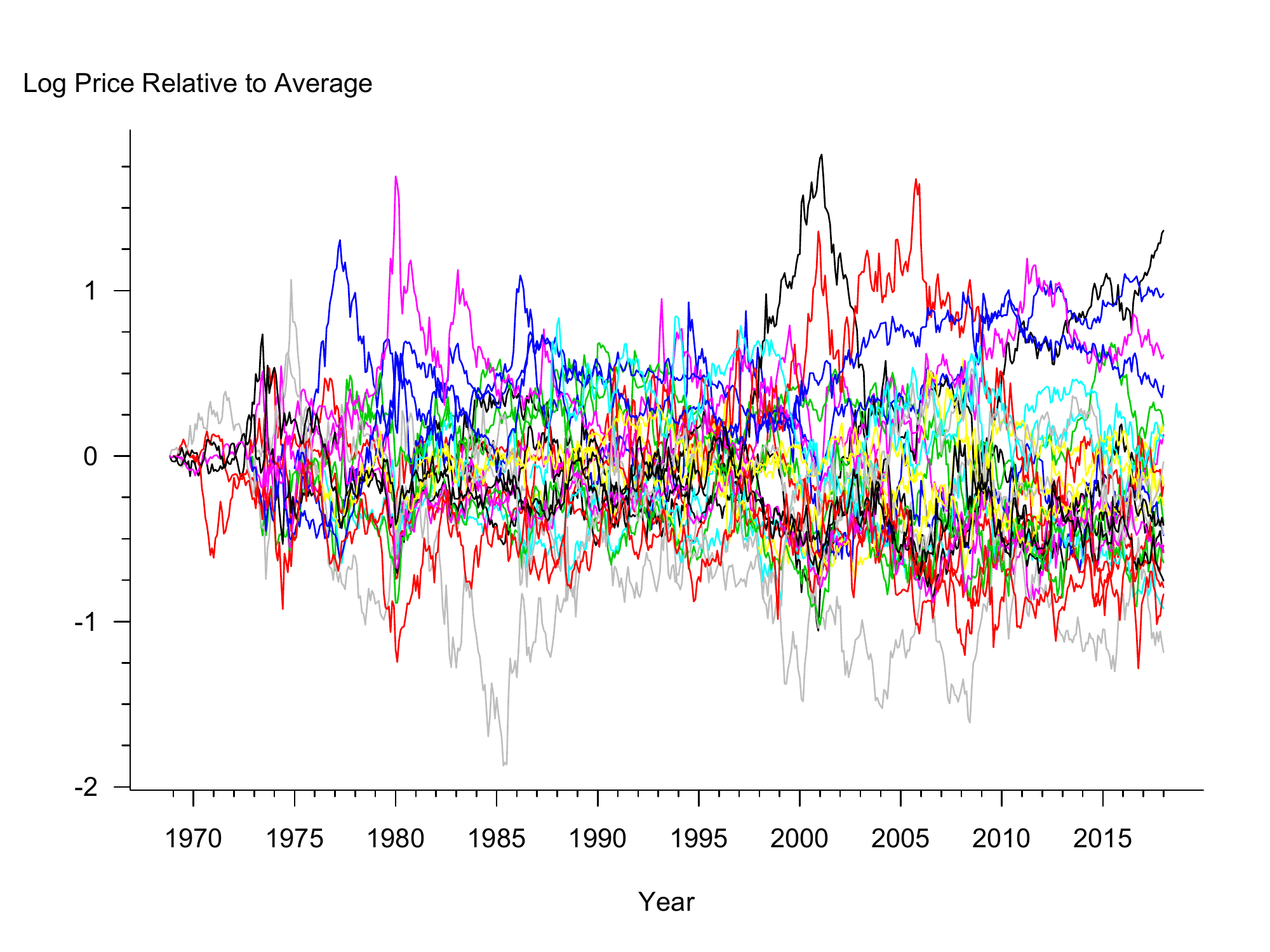}}
\end{center}
\vspace{-24pt} \caption{Two-month implied futures log-prices relative to the average, 1968-2018.}
\label{pricesFig}
\end{figure}

\begin{figure}[H]
\begin{center}
\vspace{5pt}
\hspace{-20pt}\scalebox{.67}{ \includegraphics{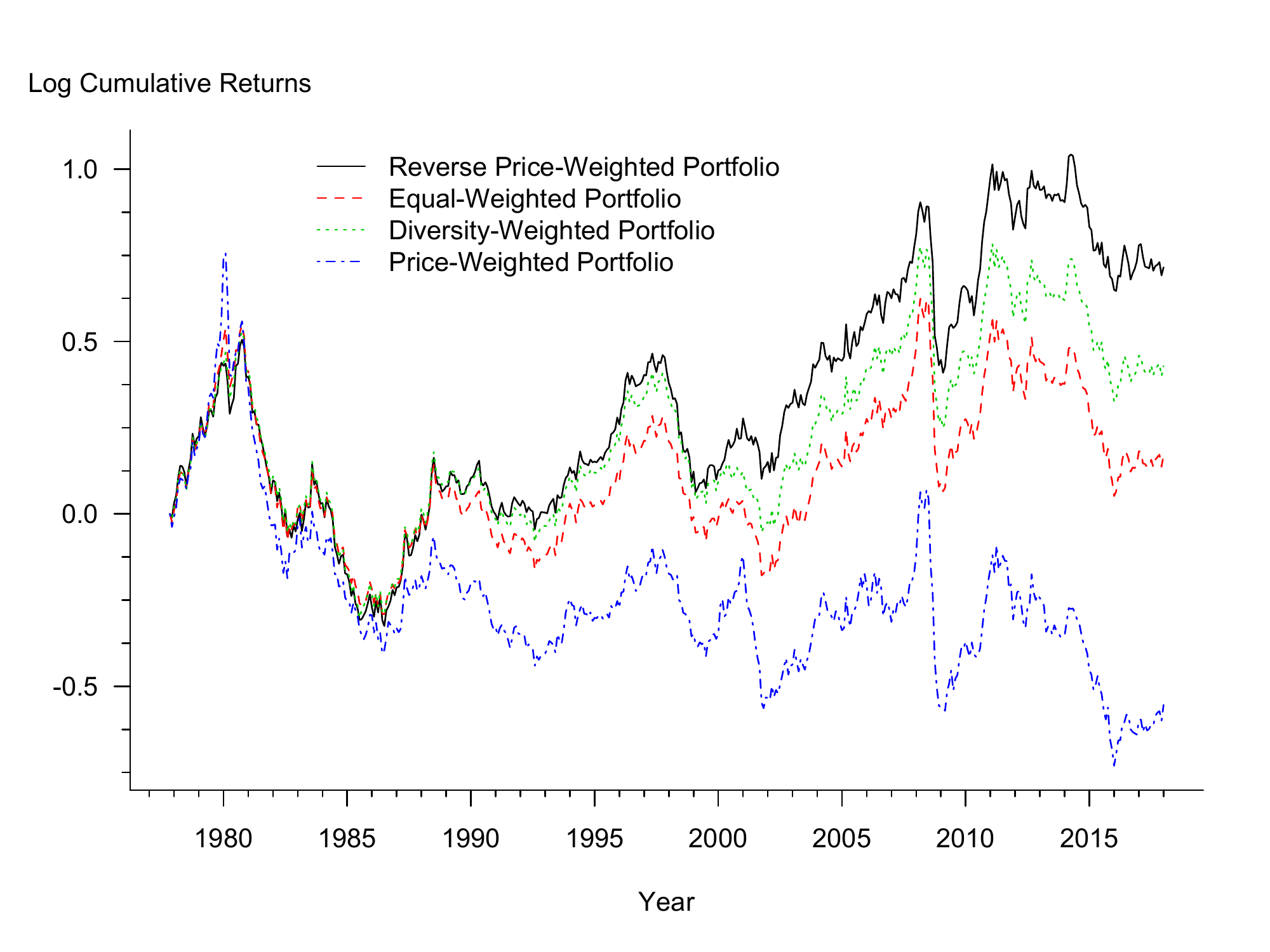}}
\end{center}
\vspace{-24pt} \caption{Cumulative log-returns for the reverse price-weighted, equal-weighted, diversity-weighted, and price-weighted portfolios, 1977-2018.}
\label{returnsFig}
\end{figure}

\begin{figure}[H]
\begin{center}
\vspace{-15pt}
\hspace{-20pt}\scalebox{.67}{ \includegraphics{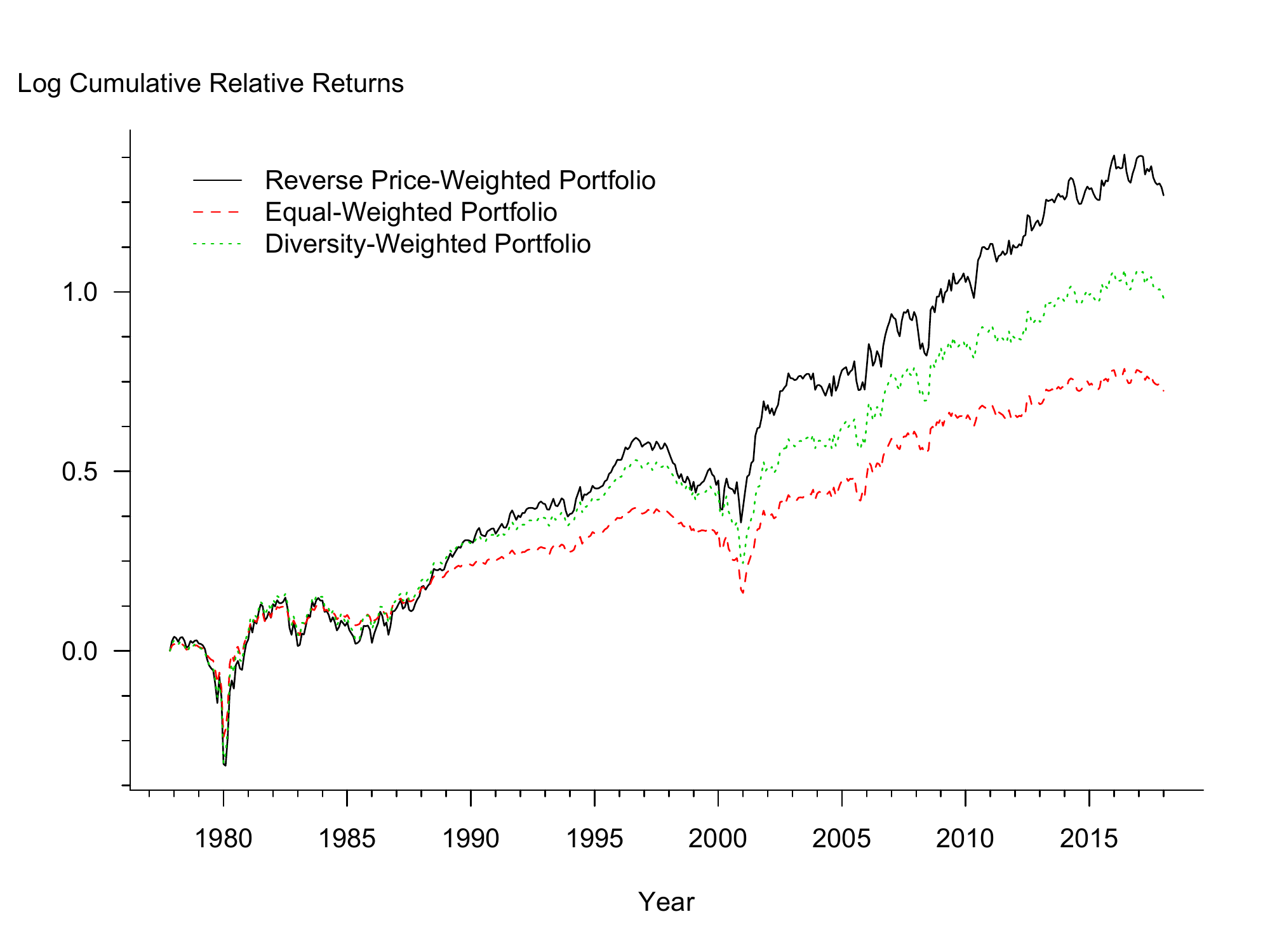}}
\end{center}
\vspace{-24pt} \caption{Cumulative log-returns for the reverse price-weighted, equal-weighted, and diversity-weighted portfolios relative to the price-weighted portfolio, 1977-2018.}
\label{relReturnsFig}
\end{figure}

\begin{figure}[H]
\begin{center}
\vspace{5pt}
\hspace{-20pt}\scalebox{.67}{ \includegraphics{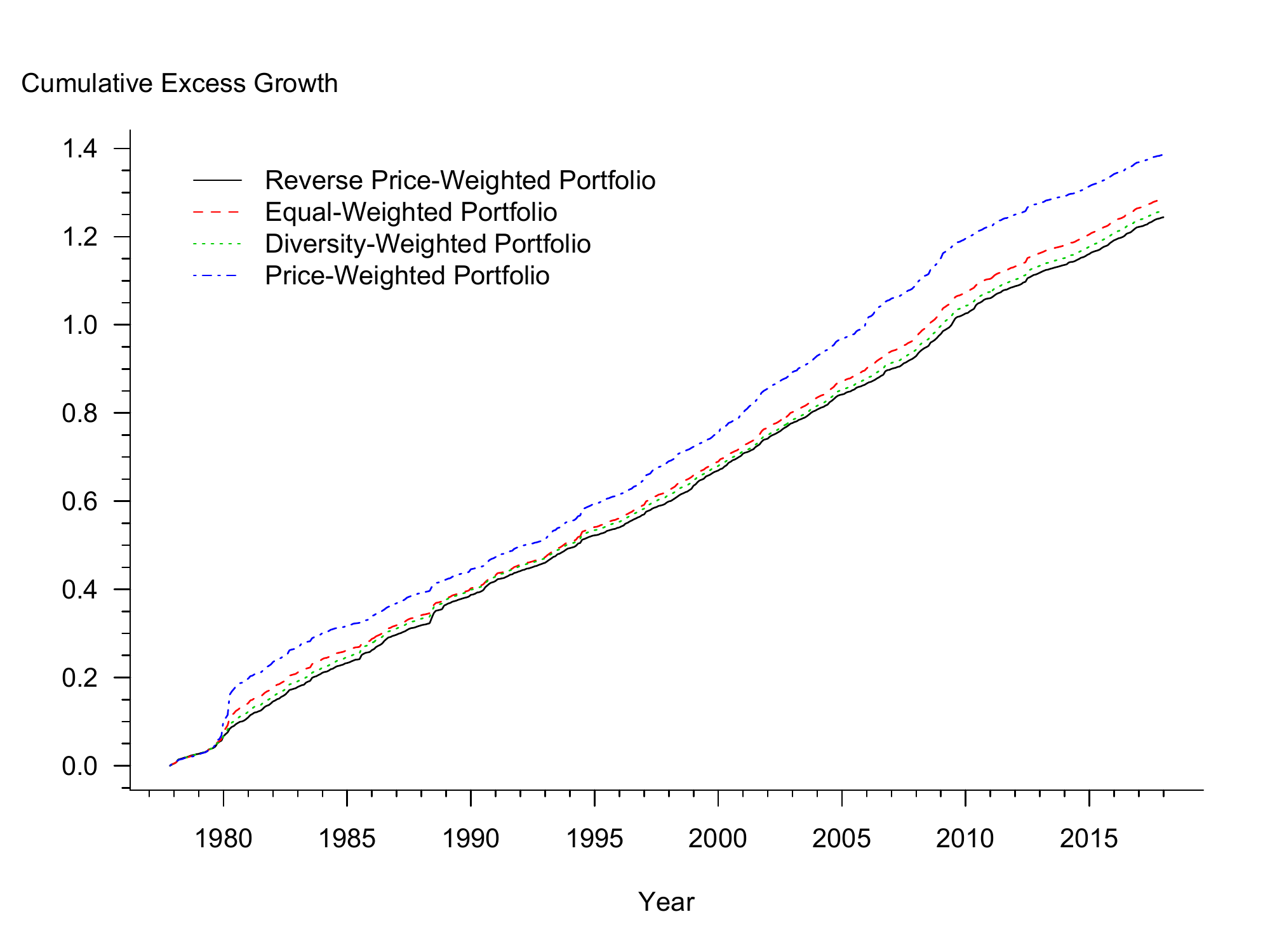}}
\end{center}
\vspace{-24pt} \caption{Cumulative values of the excess growth rate process $\g^*$ for the reverse price-weighted, equal-weighted, diversity-weighted, and price-weighted portfolios, 1977-2018.}
\label{gammaStarFig}
\end{figure}

\begin{figure}[H]
\begin{center}
\vspace{-15pt}
\hspace{-20pt}\scalebox{.67}{ \includegraphics{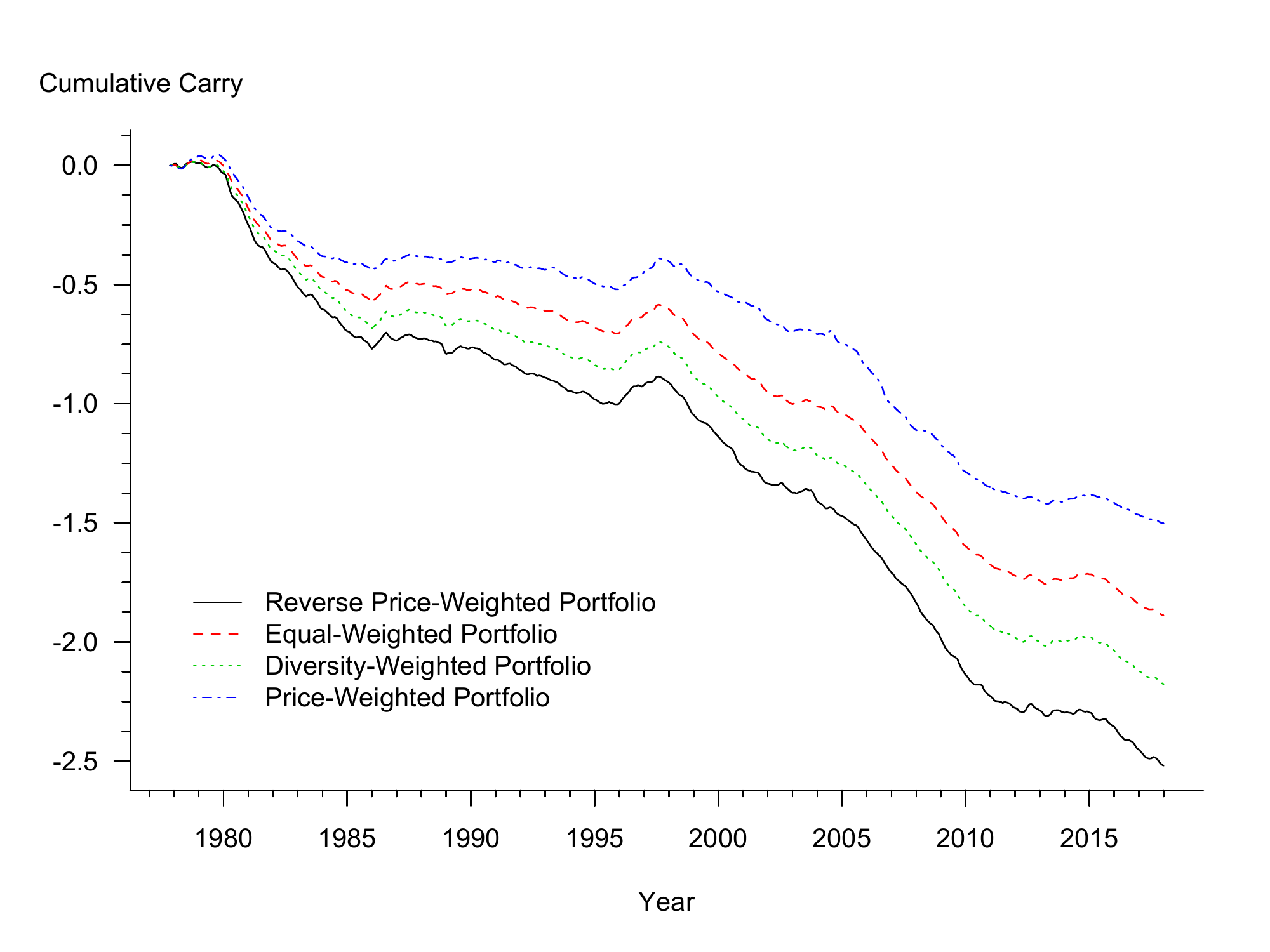}}
\end{center}
\vspace{-24pt} \caption{Cumulative carry for the reverse price-weighted, equal-weighted, diversity-weighted, and price-weighted portfolios, 1977-2018.}
\label{carryFig}
\end{figure}

\begin{figure}[H]
\begin{center}
\vspace{5pt}
\hspace{-20pt}\scalebox{.67}{ \includegraphics{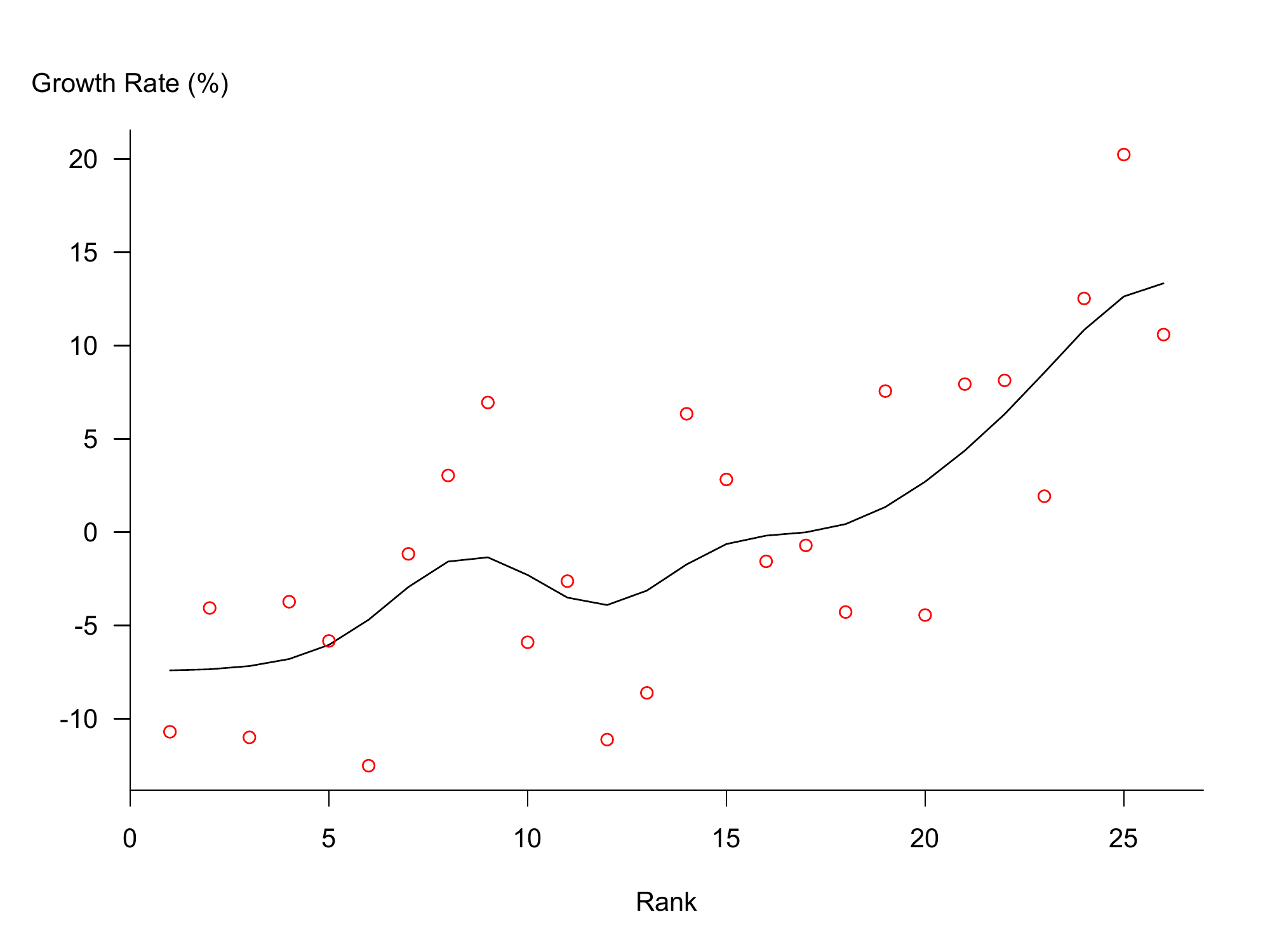}}
\end{center}
\vspace{-24pt} \caption{Estimated parameters $g_k$ for the first-order approximation of the two-month implied-futures market, 1995-2018.}
\label{GsFig}
\end{figure}

\begin{figure}[H]
\begin{center}
\vspace{-15pt}
\hspace{-20pt}\scalebox{.67}{ \includegraphics{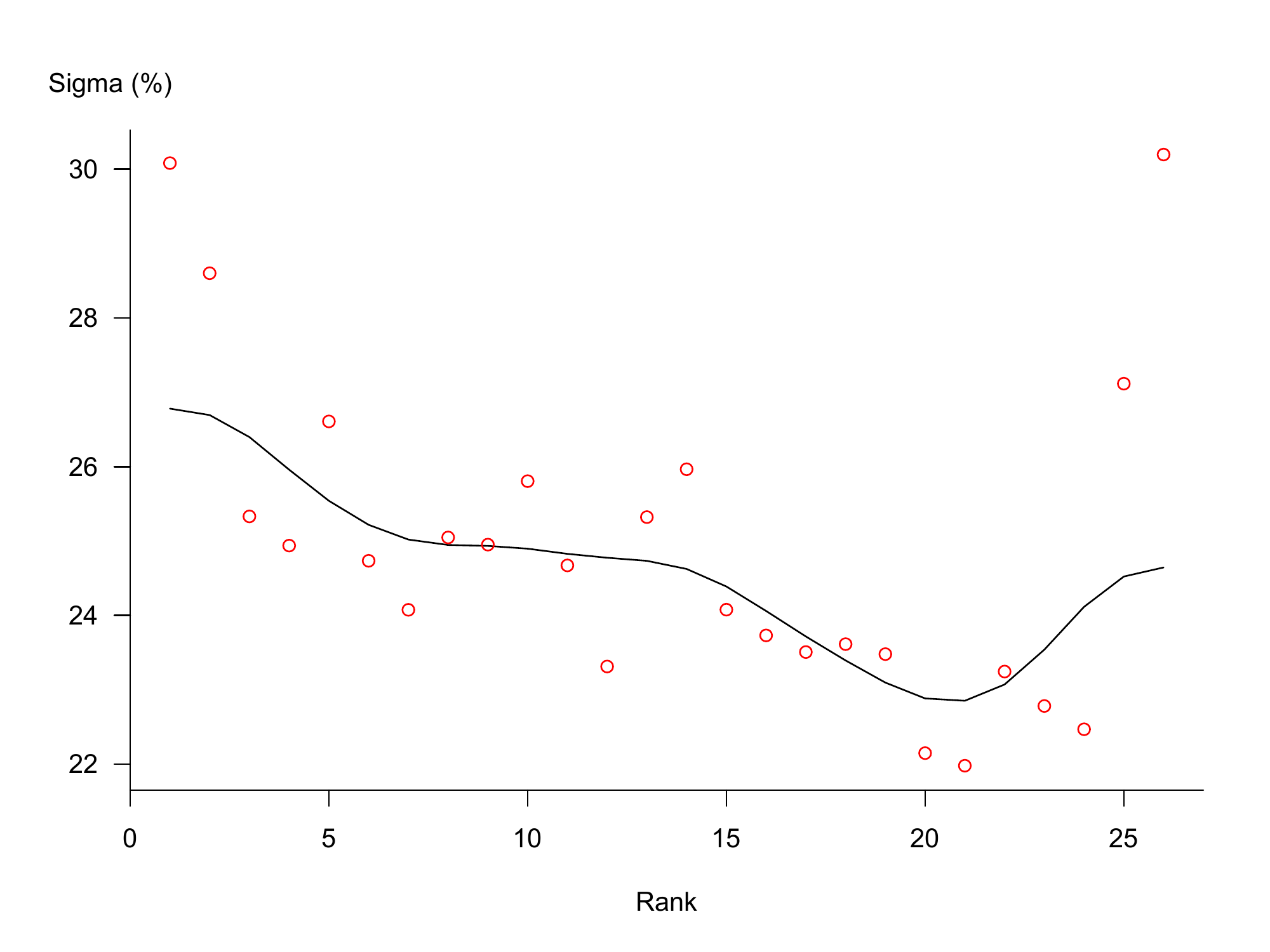}}
\end{center}
\vspace{-24pt} \caption{Estimated parameters $\s_k$ for the first-order approximation of the two-month implied-futures market, 1995-2018.}
\label{sigmasFig}
\end{figure}

\begin{figure}[H]
\begin{center}
\vspace{5pt}
\hspace{-20pt}\scalebox{.67}{ \includegraphics{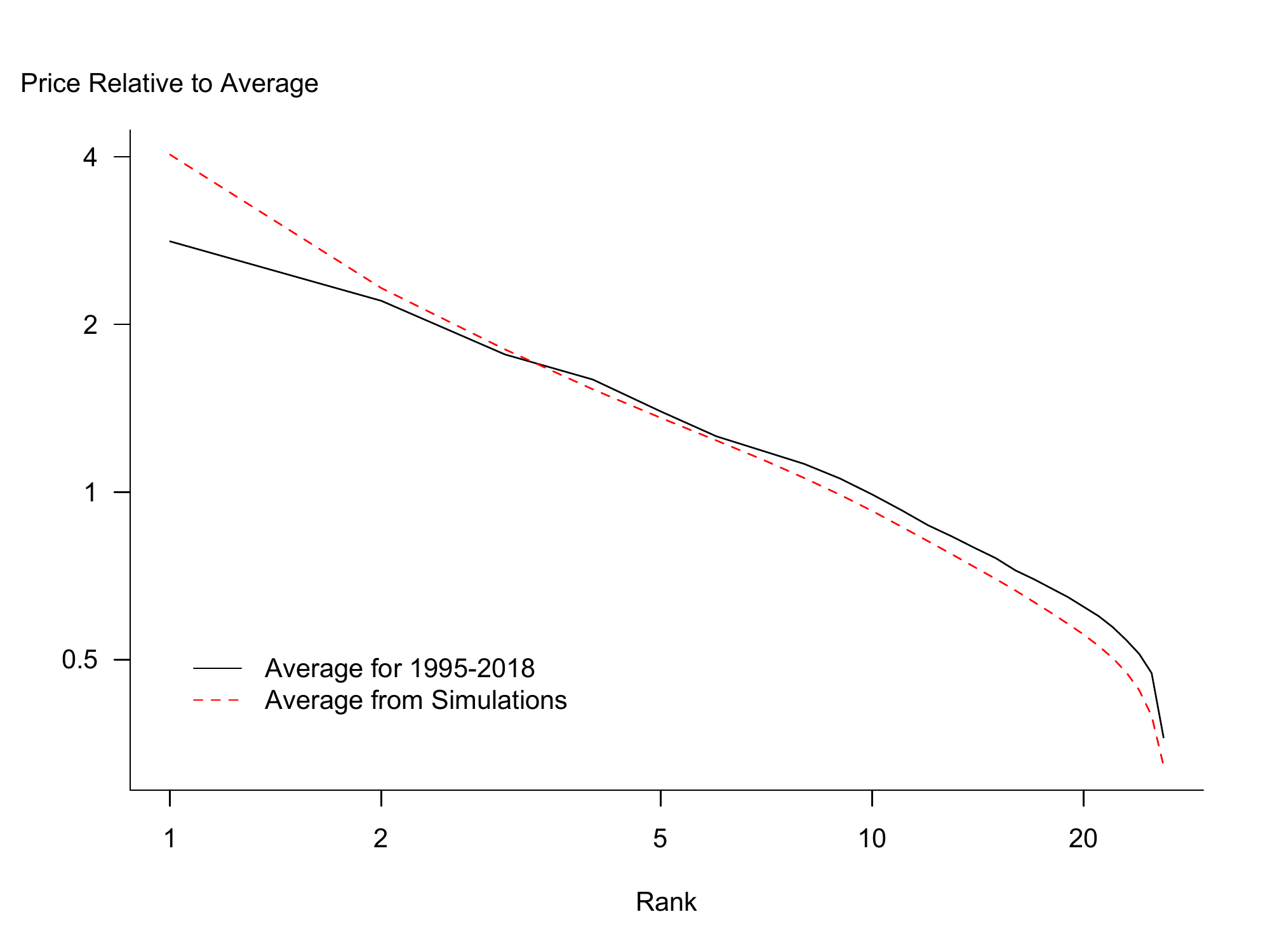}}
\end{center}
\vspace{-24pt} \caption{Average relative two-month implied futures prices for different ranks and average relative prices for different ranks from 10,000 simulations of the first-order model \eqref{4.4}, 1995-2018.}
\label{simVsDataFig}
\end{figure}

\end{document}